\documentclass[a4paper,onecolumn,11pt]{scrartcl}

\usepackage{amsmath}
\usepackage{amssymb}  
\usepackage{amsthm}
\usepackage{amsfonts}
\usepackage{graphicx}
\usepackage{bbm}
\usepackage{url}
\usepackage{authblk}
\usepackage{ upgreek }
\usepackage{dsfont}
\usepackage{makecell}
\usepackage{color}
\usepackage{subcaption}

\newcommand{\ket}[1]{| #1 \rangle}

\newcommand{\bra}[1]{\langle #1 |}

\newcommand{\proj}[2]{| #1 \rangle\!\langle #2 |}

\newcommand{\id}{\ensuremath{\mathds{1}}}

\def\beq{\begin{equation}}
\def\eeq{\end{equation}}
\def\bq{\begin{quote}}
\def\eq{\end{quote}}
\def\ben{\begin{enumerate}}
\def\een{\end{enumerate}}
\def\bit{\begin{itemize}}
\def\eit{\end{itemize}}

\def\ra{\rightarrow}

\def\lb{\left(}
\def\rb{\right)}
\def\lset{\lbrace}
\def\rset{\rbrace}
\def\lk{\left\langle}
\def\rk{\right\rangle}
\def\r|{\right|}
\def\lbr{\left[}
\def\rbr{\right]}
\def\ident{\textnormal{id}}
\def\one{\id}

\newcommand\C{\mathbbm{C}}

\newcommand\R{\mathbbm{R}}
\newcommand\N{\mathbbm{N}}
\newcommand\M{\mathcal{M}}

\newcommand{\flip}{\mathbbm{F}}

\theoremstyle{plain}
\newtheorem{thm}{Theorem}[section]
\newtheorem{lem}{Lemma}[section]
\newtheorem{cor}{Corollary}[section]

\newtheorem{defn}{Definition}[section]
\newtheorem{conj}{Conjecture}[section]
\theoremstyle{definition}

\begin{document}
\title{\vspace{-1.0cm}{\textbf{Decomposability of Linear Maps under Tensor Powers}}}

\author{Alexander M\"uller-Hermes\thanks{muellerh@posteo.net, muellerh@math.ku.dk}~}

\affil{\small{QMATH, Department of Mathematical Sciences, University of Copenhagen, Universitetsparken 5, 2100 Copenhagen, Denmark}}

\maketitle
\date{\today}

\begin{abstract}

Both completely positive and completely copositive maps stay decomposable under tensor powers, i.e~under tensoring the linear map with itself. But are there other examples of maps with this property? We show that this is not the case: Any decomposable map, that is neither completely positive nor completely copositive, will lose decomposability eventually after taking enough tensor powers. Moreover, we establish explicit bounds to quantify when this happens. To prove these results we use a symmetrization technique from the theory of entanglement distillation, and analyze when certain symmetric maps become non-decomposable after taking tensor powers. Finally, we apply our results to construct new examples of non-decomposable positive maps, and establish a connection to the PPT squared conjecture. 

\end{abstract}

\tableofcontents

\section{Introduction and main results}

Let $\M_d$ denote the set of complex $d\times d$ matrices, and $\M^+_d$ the cone of positive semidefinite matrices (in the following simply called ``positive matrices''). A linear map $P:\M_{d_A}\ra \M_{d_B}$ for dimensions $d_A,d_B\in \N$ is called \emph{positive} iff $P(\M^+_{d_A})\subseteq\M^+_{d_B}$. As in \cite{muller2016positivity} we will call a linear map $P$ \emph{n-tensor-stable positive} for $n\in\N$ iff $P^{\otimes n}$ is positive, and \emph{tensor-stable positive} iff $P^{\otimes n}$ is positive for any $n\in\N$. Examples of tensor-stable positive maps $P:\M_{d_A}\ra \M_{d_B}$ are the \emph{completely positive maps}, i.e. where $(\ident_k\otimes P)\geq 0$ for any $k\in\N$, and the \emph{completely copositive maps}, i.e. where $(\vartheta_k\otimes P)\geq 0$. Here, $\vartheta_k:\M_k\ra\M_k$ denotes the matrix transposition with respect to some fixed basis. The central question of \cite{muller2016positivity} is whether these are the only classes of tensor-stable positive maps. This problem is still unsolved. 

A linear map $P:\M_{d_A}\ra \M_{d_B}$ is called \emph{decomposable} iff it is of the form $P = T_1 + \vartheta_{d_2}\circ T_2$ for completely positive maps $T_1,T_2 :\M_{d_A}\ra\M_{d_B}$. Decomposable maps are clearly positive, and there are many examples of such maps that are neither completely positive nor completely copositive. It is therefore a natural question, how decomposability behaves under tensor products. We make the following definition:   

\begin{defn}[Tensor-stable decomposability]
A linear map $P:\M_{d_A}\ra\M_{d_B}$ is called 
\begin{itemize}
\item \emph{n-tensor-stable decomposable} for $n\in\N$ iff $P^{\otimes n}$ is decomposable. 
\item \emph{tensor-stable decomposable} iff $P^{\otimes n}$ is decomposable for any $n\in \N$.
\end{itemize}
\end{defn}  
The notion of $2$-tensor-stable decomposability has been first considered in \cite{filippov2017positive}, where some explicit examples for $d_A=d_B=2$ have been studied. Our main result is 

\begin{thm}
If $P:\M_{d_A}\ra\M_{d_B}$ is tensor-stable decomposable, then $P$ is completely positive or completely copositive. 
\label{thm:NoTensStabDec}
\end{thm} 

We will give two different proofs for this result: One qualitative proof based on ideas from entanglement distillation, and a quantitative proof giving explicit bounds on $n\in\N$ such that a given decomposable map is $n$-tensor-stable decomposable. These bounds will be formulated in terms of the following quantity. 

\begin{defn}
Given a non-zero positive map $P:\M_{d_A}\ra\M_{d_B}$ we define 
\[
\mu\lb P\rb := \inf\Big\lset \frac{\text{Tr}\lbr C_P C_T\rbr}{\text{Tr}\lbr (T^{*}\circ P)(\one_{d_A})\rbr} ~\Big{|}~ T:\M_{d_A}\ra\M_{d_B} \text{completely positive}, (T^*\circ P)(\one_{d_A})\neq 0\Big\rset
\]
where $C_P$ and $C_T$ denote the Choi matrices of $P$ and $T$ respectively, see \eqref{equ:ChoiMat}, and $T^*$ denotes the adjoint of $T$ with respect to the Hilbert-Schmidt inner product\footnote{$\lk A,B\rk := \text{Tr}\lbr A^\dagger B\rbr$ for matrices $A,B\in\M_d$}.
\label{defn:mu}
\end{defn}

We will see in Section \ref{sec:Symmetrizing} that $\mu\lb P\rb$ is defined for any non-zero positive map $P:\M_{d_A}\ra\M_{d_B}$, and that $\mu\lb P\rb\in \lbr -1, 1/{d_A}\rbr$ , with $\mu\lb P\rb\geq 0$ iff $P$ is completely positive. We will show

\begin{thm}[Quantitative bound]
Let $P:\M_{d_A}\ra\M_{d_B}$ be a non-zero positive map, and set $d:=\min(d_A,d_B)$. If $P$ is $n$-tensor-stable decomposable, then 
\[
\max\lb\mu\lb P\rb,\mu\lb\vartheta_{d_B}\circ P\rb\rb\geq \lb\frac{2}{1+\sqrt[\left\lfloor \frac{n}{2}\right\rfloor~]{1+\sqrt{\frac{2d}{d+1}}}} - 1\rb,
\]
where $\lfloor x\rfloor$ denotes the largest integer less than $x$. 
\label{thm:QuantBound}
\end{thm}

In the limit $n\ra \infty$ the previous theorem shows that any tensor-stable decomposable map $P$ satisfies $\max\lb\mu\lb P\rb,\mu\lb\vartheta_{d_B}\circ P\rb\rb\geq 0$ implying that $P$ is completely positive or completely copositive. This proves Theorem \ref{thm:NoTensStabDec}.

To prove Theorem \ref{thm:QuantBound} we will use a symmetrization technique from \cite{muller2016positivity} to reduce the problem to mixed tensor products of Werner maps (see Section \ref{sec:Symmetrizing}). Using the symmetries of these maps we characterize the parameter regions where they are decomposable via a linear program (see Section \ref{sec:MixedTensroHW} and in particular Section \ref{sec:DecompHW}). In simple cases (see Section \ref{sec:SpecialCase}) this linear program can be solved exactly. In general (see Section \ref{sec:AnalytBoundHW}) we can find a specific feasible point giving an analytical bound, and leading to proofs for Theorem \ref{thm:QuantBound} and Theorem \ref{thm:NoTensStabDec} (see Section \ref{sec:Proofs}). Finally, in Section \ref{sec:Appl} we apply our results to construct new examples of non-decomposable positive maps, and to establish an implication of tensor-stable positive maps to the so called PPT squared conjecture~\cite{christandl2012PPT}.

\section{Notation and preliminaries}

We will denote by $\one_d\in\M_d$ the $d\times d$ unit matrix, by $\omega_{d}\in(\M_{d}\otimes \M_d)^+$ the (unnormalized) maximally entangled state, i.e.~$\omega_d=\proj{\Omega_d}{\Omega_d}$ for $\ket{\Omega_d}=\sum^d_{i=1}\ket{i}\otimes \ket{i}\in \C^d\otimes \C^d$, and by $\flip_d\in\M_d\otimes \M_d$ the flip operator defined by $\flip_d(\ket{i}\otimes \ket{j})=\ket{j}\otimes \ket{i}$. Here $\lset \ket{i}\rset^d_{i=1}\subset \C^d$ denotes the computational basis, i.e.~the vector $\ket{i}$ has a single $1$ in the $i$th position and zeros in the remaining entries. For $n,m\in \N$ we will denote by $\M_{n,m}$ the set of complex rectangular $n\times m$-matrices, and by $\M_{n,m}\lb\R\rb$ the subset of $n\times m$-matrices with real entries.

We denote by $\ident_d:\M_d\ra\M_d$ the identity map $\ident_d(X)=X$ and by $\vartheta_d:\M_d\ra\M_d$ the matrix transposition $\vartheta_d(X) = X^T$ in the computational basis (our results will not depend on this choice of basis). Given an operator $Y:\C^{d_A}\ra\C^{d_B}$ we denote by $\text{Ad}_Y:\M_{d_A}\ra\M_{d_B}$ the completely positive map $\text{Ad}_Y(X) = YXY^\dagger$. The \emph{partial transposition} $(\ident_{d_A}\otimes \vartheta_{d_B}):\M_{d_A}\otimes \M_{d_B}\ra\M_{d_A}\otimes \M_{d_B}$ will play an important role in the following. To simplify formulas we will sometimes write $X^{\Gamma}:= (\ident_{d_A}\otimes \vartheta_{d_B})(X)$ to abbreviate the partial transposition (on the second tensor factor) of a bipartite matrix $X\in \M_{d_A}\otimes \M_{d_B}$. Here the dimensions will be clear from context.    

\subsection{Choi-Jamiolkowski isomorphism and mapping cones}

The Choi-Jamiolkowski isomorphism~\cite{choi1975completely,jamiolkowski1972linear} relates each linear map $L:\M_{d_A}\ra\M_{d_B}$ to its Choi matrix 
\begin{equation}
C_L = (\ident_{d_A}\otimes L)(\omega_{d_A})\in \M_{d_A}\otimes \M_{d_B}.
\label{equ:ChoiMat}
\end{equation}
Under the Choi-Jamiolkowski isomorphism, positive maps $P:\M_{d_A}\ra\M_{d_B}$ correspond to \emph{block positive} matrices $C_P\in\M_{d_A}\otimes\M_{d_B}$, i.e. such that $(\bra{x}\otimes \bra{y}) C_P (\ket{x}\otimes \ket{y})\geq 0$ holds for any $\ket{x}\in\C^{d_A}$ and any $\ket{y}\in\C^{d_B}$. Similarly, completely positive maps $T:\M_{d_A}\ra\M_{d_B}$ correspond to positive matrices $C_T\in(\M_{d_A}\otimes \M_{d_B})^+$. 

Let $\mathcal{P}\lb n,m\rb$ denote the cone of positive maps $P:\M_n\ra\M_{m}$. The notion of mapping cones was introduced by E.~St{\o}rmer in \cite{stormer1986extension} (see also \cite{skowronek2009cones} for more details). The following is a slight modification of the original definition: 

\begin{defn}[Mapping cones]
We call a system $\mathcal{C} = \lset \mathcal{C}_{n,m}\rset_{n,m\in \N}$ of subcones $\mathcal{C}_{n,m}\subset \mathcal{P}\lb n,m\rb$ a \emph{mapping cone} if the following conditions are satisfied:
\begin{enumerate}
\item For any $n,m\in\N$ the subcone $\mathcal{C}_{n,m}$ is closed. 
\item For any $n,m,n',m'\in \N$, $P\in \mathcal{C}_{n,m}$ and completely positive maps $T:\M_{n'}\ra\M_n$ and $S:\M_{m}\ra\M_{m'}$ we have that $S\circ P\circ T\in \mathcal{C}_{n',m'}$.
\end{enumerate}
For $P:\M_{n}\ra\M_{m}$ we will simply write $P\in \mathcal{C}$ instead of $P\in \mathcal{C}_{n,m}$.
\label{defn:MappingCone}
\end{defn}

In the following we will focus mostly on the cones of positive maps and of decomposable maps, and we refer to \cite{skowronek2009cones} for more examples of mapping cones. It is sometimes convenient to characterize mapping cones via their dual cones. The following two paragraphs implicitly contain examples of this, but we will not go into further details here. 

A positive matrix $X\in (\M_{d_A}\otimes \M_{d_B})^+$ is called \emph{separable} iff there exists a $k\in \N$ such that 
\[
X = \sum^k_{i=1} Y_i\otimes Z_i,
\]
with $Y_i\in (\M_{d_A})^+$ and $Z_i\in (\M_{d_B})^+$ for any $i\in\lset 1,\ldots ,k\rset$. Via the Choi-Jamiolkowski isomorphism separable matrices correspond to so called \emph{entanglement breaking} maps (see \cite{horodecki2003entanglement}), i.e. completely positive maps $T:\M_{d_A}\ra\M_{d_B}$ such that $(\ident_{d_A}\otimes T)(X)$ is separable for any positive matrix $X\in (\M_{d_A}\otimes \M_{d_A})^+$. The mapping cone of entanglement breaking maps is dual to the mapping cone of positive maps (see \cite{skowronek2009cones} for details), which equivalently means that a positive matrix $X\in (\M_{d_A}\otimes \M_{d_B})^+$ is separable iff $(\ident_{d_A}\otimes P)\lb X\rb\geq 0$ for any positive map $P:\M_{d_B}\ra\M_{d_A}$. Conversely, a linear map $L:\M_{d_A}\ra\M_{d_B}$ is positive iff $(\ident_{d_B}\otimes L)\lb X\rb\geq 0$ for any separable matrix $X\in(\M_{d_B}\otimes \M_{d_A})^+$.

In the cases $d_Ad_B\leq 6$ any positive linear map $P:\M_{d_A}\ra\M_{d_B}$ is decomposable (see \cite{woronowicz1976positive}). For higher dimensions this is no longer true, and the structure of positive maps seems to be very complicated. Contrary to this, the set of decomposable maps behaves more nicely. The mapping cone dual to the mapping cone of decomposable maps is given by the linear maps that are both completely positive and completely copositive. Via the Choi-Jamiolkowski isomorphism these correspond to positive matrices with positive partial transpose. This duality can be expressed as follows:

\begin{thm}[Testing decomposability~\cite{stormer1982decomposable}]
For any linear map $L:\M_{d_A}\ra\M_{d_B}$ the following are equivalent:
\begin{enumerate}
\item $L$ is decomposable.
\item $(\ident_{d_A}\otimes L)(X)\geq 0$ for any $X\in (\M_{d_A}\otimes \M_{d_B})^+$ with $X^{\Gamma}\geq 0$.
\item $\text{Tr}\lb X C_L\rb\geq 0$ for any $X\in (\M_{d_A}\otimes \M_{d_B})^+$ with $X^{\Gamma}\geq 0$.
\end{enumerate}
\label{thm:DecWitness}
\end{thm}

Note that the last condition in Theorem \ref{thm:DecWitness} can be checked efficiently via semidefinite programming~\cite{boyd2004convex}.

\subsection{Twirling, Werner states, and Werner maps}
\label{sec:WernerStatesAndMaps}

The projectors onto the symmetric and antisymmetric subspaces of $\C^{d}\otimes \C^{d}$ are given by
\[
P_{\text{sym}} = \frac{1}{2} (\one_d\otimes \one_d + \mathbbm{F}_d), \hspace{1cm} P_{\text{asym}} = \frac{1}{2}(\one_d\otimes \one_d - \mathbbm{F}_d).
\]
Let $\mathcal{U}_d\subset \M_d$ denote the set of $d\times d$ unitary matrices. The \emph{$UU$-twirl} denoted by
\[
T_{UU}:\M_d\otimes \M_d\ra\M_d\otimes \M_d
\] 
is defined as 
\begin{equation}
T_{UU}(X) = \int_{U\in\mathcal{U}_d} (U\otimes U)X(U\otimes U)^\dagger \text{d}U = \text{Tr}\lbr X P_{\text{sym}}\rbr\frac{P_{\text{sym}}}{d_{\text{sym}}} + \text{Tr}\lbr X P_{\text{asym}}\rbr\frac{P_{\text{asym}}}{d_{\text{asym}}}
\label{equ:Twirl}
\end{equation}
for any $X\in\M_d\otimes \M_d$, where $d_{\text{sym}}:=\text{Tr}\lbr P_\text{sym}\rbr = d(d+1)/2$ and $d_{\text{asym}}:=\text{Tr}\lbr P_\text{asym}\rbr = d(d-1)/2$. Here the integration is with respect to the Haar measure on the unitary group $\mathcal{U}_d$, and the second equality follows from the Schur-Weyl duality (see \cite{werner1989quantum} for details).

Under the $UU$-twirl every quantum state, i.e.~a positive matrix with unit trace, gets mapped to the family of \emph{Werner states} (see \cite{werner1989quantum}) given by
\begin{equation}
\rho_{W}(p) = p\frac{P_{\text{sym}}}{d_{\text{sym}}} + (1-p)\frac{P_{\text{asym}}}{d_{\text{asym}}} \in (\M_d\otimes \M_d)^+
\label{equ:WernerStates}
\end{equation} 
with the parameter $p\in\lbr 0, 1\rbr$. It is well-known that the Werner state $\rho_{W}(p)$ is \emph{separable} iff it has positive partial transpose, which holds iff $p\geq 1/2$. 

Via the Choi-Jamiolkowski isomorphism we can relate the Werner states $\rho_{W}(p)\in (\M_d\otimes \M_d)^+$ to the family of \emph{Werner maps} denoted by $W_p:\M_d\ra \M_d$ via $C_{W_p} = \rho_{W}(p)$. The maps in this family are sometimes called \emph{depolarized Werner-Holevo maps} after their most prominent example, the Werner-Holevo channel $dW_0$, studied in~\cite{werner2002counterexample}. However, since they are equivalent to the Werner states, and we are not using the Werner-Holevo channel in particular, we decided for the aforementioned name. 

In the following we will often consider \emph{mixed tensor powers of Werner maps} given by 
\[
W^{\otimes n}_{p_1}\otimes (\vartheta_{d}\circ W_{p_2})^{\otimes m}:\M^{\otimes (n+m)}_{d}\ra \M^{\otimes (n+m)}_{d}
\] 
for $p_1,p_2\in\lbr 0,1\rbr$.

\section{Symmetrizing positive maps}
\label{sec:Symmetrizing}

The symmetrization techniques introduced in this section originate from the theory of entanglement distillation (see for example \cite{horodecki1999reduction}) and have been adapted to the study of positive maps in \cite{muller2016positivity}. We present these techniques here in a slightly more general form mostly to make this article self-contained, but also to make them more applicable for further studies.  

 Via the Choi-Jamiolkowski isomorphism we can identify positive maps $P:\M_{d_A}\ra\M_{d_B}$ with block positive matrices $C_P\in \M_{d_A}\otimes \M_{d_B}$ and vice-versa. The following lemma is probably well-known.

\begin{lem}[Operator inequality for block positive matrices]
Any positive map $P:\M_{d_A}\ra\M_{d_B}$ satisfies
\[
-\one_{d_A}\otimes P(\one_{d_A})\leq C_P\leq d_A\one_{d_A}\otimes P(\one_{d_A}).
\]

\label{lem:OpInequ}
\end{lem} 
\begin{proof}
For the first inequality note that $2P^{\Gamma}_{\text{sym}} = \one_{d_A}\otimes \one_{d_A}+\omega_{d_A}$ is separable. Then, by positivity of $P$ we have
\[
\one_{d_A}\otimes P(\one_{d_A}) + C_P=(\ident_{d_A}\otimes P)(2P^{\Gamma}_{\text{sym}})\geq 0.
\]
Similarly, since $X = \one_{d_A}\otimes \one_{d_A}-\omega_{d_A}/d_A$ is separable (note that $(d_A^2-1)X=\rho_W\lb 1/2\rb^\Gamma$) we have that
\[
\one_{d_A}\otimes P(\one_{d_A}) - C_P/d_A=(\ident_{d_A}\otimes P)(X)\geq 0.
\] 
\end{proof}

Recall the quantity $\mu(P)$ from Definition \ref{defn:mu} that was associated with a non-zero positive map $P:\M_{d_A}\ra\M_{d_B}$. Note that if $P$ is unital, then $\mu(P)$ coincides with the minimal eigenvalue of the Choi matrix $C_P$. For general $P$ a similar interpretation is possible in terms of the positive map $P'=\text{Ad}_{X}\circ P$ for $X = P(\one_d)^{-1/2}$ with a pseudoinverse (see e.g.~\cite{penrose1955generalized}) in case $P(\one_d)$ is not full rank, and possibly restricting the image to make $P'$ unital. However, since this alternative form of $\mu(P)$ does not seem to give more insights, we will not go into details of this, and just stick to Definition \ref{defn:mu}.    

The previous lemma leads to upper and lower bounds on this quantity:

\begin{lem}
For any non-zero positive map $P:\M_{d_A}\ra\M_{d_B}$ we have $\mu\lb P\rb\in \lbr -1, \frac{1}{d_A}\rbr$. Moreover, $P$ is completely positive iff $\mu\lb C_P\rb\geq 0$.
\label{lem:BasicMu}
\end{lem} 
\begin{proof}
For any positive map $P:\M_{d_A}\ra\M_{d_B}$ we can consider the completely positive map $T:\M_{d_A}\ra\M_{d_B}$ given by $T(X)=\text{Tr}\lbr X\rbr P(\one_{d_A})$. This map satisfies $(T^*\circ P)(\one_{d_A})\neq 0$ and is therefore contained in the set used to define $\mu$ in Definition \ref{defn:mu}. Evaluating this example shows that $\mu\lb P\rb\leq 1/d_A$. The lower bound $-1\leq \mu\lb P\rb$ follows immediately from the first inequality in Lemma \ref{lem:OpInequ} using that the Choi matrix of a completely positive map is positive. 

Clearly, for completely positive $P$ we have $C_P\geq 0$ and thus $\mu\lb P\rb\geq 0$. If $C_P\ngeq 0$, then by the Choi-Jamiolkowski isomorphism there exists a completely positive map $T:\M_{d_A}\ra\M_{d_B}$ such that $\text{Tr}\lbr C_P C_T\rbr<0$. By Lemma \ref{lem:OpInequ} we find that 
\[
-\text{Tr}\lbr (T^{*}\circ P)(\one_{d_A})\rbr=-\text{Tr}\lbr P(\one_{d_A})T(\one_{d_A})\rbr = -\text{Tr}\lbr (\one_{d_A}\otimes P(\one_{d_A}))C_T\rbr\leq \text{Tr}\lbr C_PC_T\rbr<0,
\]
and thereby $(T^*\circ P)(\one_{d_A})\neq 0$ showing that $T$ is contained in the set used to define $\mu$ in Definition \ref{defn:mu}. 
\end{proof}

We can now state our symmetrization theorem:

\begin{thm}[Symmetrization of positive maps]
Let $P:\M_{d_A}\ra \M_{d_B}$ be positive, but not completely positive. Then, for any $p$ satisfying
\begin{equation}
0\leq \frac{1}{2}\lb 1 + \mu\lb P\rb\rb < p < \frac{1}{2},
\label{equ:pInterval}
\end{equation}
there exists a completely positive map $S:\M_{d_B}\ra\M_{d_A}$ such that 
\[
\int_{\mathcal{U}_d} \text{Ad}_U\circ \vartheta_{d_A}\circ S\circ P\circ \text{Ad}_{U^T} \text{d}U = W_p,
\]
for the Werner map $W_p:\M_{d_A}\ra\M_{d_A}$. 

\label{thm:SymmP}
\end{thm}
\begin{proof}
Consider $p$ such that \eqref{equ:pInterval} holds (this interval is not empty by Lemma \ref{lem:BasicMu}). From Definition \ref{defn:mu} it is easy to see that there exists a completely positive map $T:\M_{d_B}\ra\M_{d_A}$ such that 
\[
p = \frac{1}{2}\lb 1 + \frac{\text{Tr}\lbr C_P C_{T^*}\rbr}{\text{Tr}\lbr (T\circ P)(\one_{d_A}) \rbr}\rb <\frac{1}{2},
\]
with $T^*$ denoting the adjoint of $T$ with respect to the Hilbert-Schmidt inner product. Now, an elementary computation using $\omega^{\Gamma}_{d_A} = \flip_{d_A}$ and $\vartheta_{d_A}\circ \vartheta_{d_A}=\ident_{d_A}$ shows that 
\[
\text{Tr}\lbr C_P C_{T^*}\rbr = \text{Tr}\lbr C_{T\circ P}\omega_{d_A}\rbr = \text{Tr}\lbr C_{\vartheta_{d_A}\circ T\circ P}\flip_{d_A}\rbr.
\]
Furthermore, note that by positivity of $T\circ P\neq 0$ we have
\[
0<\text{Tr}\lbr (T\circ P)(\one_{d_A}) \rbr = \text{Tr}\lbr (\vartheta_{d_A}\circ T\circ P)(\one_{d_A})\rbr = \text{Tr}\lbr C_{\vartheta_{d_A}\circ T\circ P}\rbr.
\]
Combining the previous equalities shows that 
\[
p = \text{Tr}\lbr \frac{C_{\vartheta_{d_A}\circ T\circ P}}{\text{Tr}\lbr C_{\vartheta_{d_A}\circ T\circ P}\rbr} P_{\text{sym}}\rbr.
\]
Since $0\leq p<1/2$ and $P_{\text{sym}}+P_{\text{asym}} = \one_{d_A}\otimes \one_{d_A}$ we also have that  
\[
0\leq \text{Tr}\lbr \frac{C_{\vartheta_{d_A}\circ T\circ P}}{\text{Tr}\lbr C_{\vartheta_{d_A}\circ T\circ P}\rbr} P_{\text{asym}}\rbr\leq 1.
\]
Now, we can apply the $UU$-twirl from \eqref{equ:Twirl} and obtain
\[
\int_{U\in\mathcal{U}_{d_A}} (U\otimes U)\frac{C_{\vartheta_{d_A}\circ T\circ P}}{\text{Tr}\lbr C_{\vartheta_{d_A}\circ T\circ P}\rbr}(U\otimes U)^\dagger \text{d}U = \rho_W(p).
\]
Finally, note that for each $U\in \mathcal{U}_{d_A}$ 
\[
(U\otimes U)C_{\vartheta_{d_A}\circ T\circ P}(U\otimes U)^{\dagger} = C_{\text{Ad}_{U}\circ \vartheta_{d_A}\circ T\circ P\circ \text{Ad}_{U^T}}.
\]
Using the Choi-Jamiolkowski isomorphism to express the previous equations in form of linear maps and setting $S = T/\text{Tr}\lbr C_{\vartheta_{d_A}\circ T\circ P} \rbr$ finishes the proof.

\end{proof}

Note that by applying the previous theorem for the dual map $P^*:\M_{d_B}\ra\M_{d_A}$ and taking the dual afterwards leads to a similar statement with a Werner map $W_p:\M_{d_B}\ra\M_{d_B}$ for $p$ with $\frac{1}{2}\lb 1 + \mu\lb P^*\rb\rb < p < \frac{1}{2}$. 

Now we can proof the following symmetrization theorem.

\begin{thm}
Let $\mathcal{C}$ denote a mapping cone according to Definition \ref{defn:MappingCone}, and $P:\M_{d_A}\ra \M_{d_B}$ be a positive map that is neither completely positive nor completely copositive. If for some $N\in\N$ we have $P^{\otimes N}\in \mathcal{C}$, then for any $n,m\in\N$ with $n+m\leq N$, and any $p_1,p_2$ with
\begin{equation}
\frac{1}{2}\lb 1 + \mu\lb \vartheta_{d_B}\circ P\rb\rb < p_1 < \frac{1}{2}\hspace*{0.5cm}\text{ and }\hspace*{0.5cm} \frac{1}{2}\lb 1 + \mu\lb P\rb\rb < p_2 < \frac{1}{2},
\label{equ:p1p2RangesFirst}
\end{equation}
the Werner maps $W_{p_1}, W_{p_2}:\M_{d_A}\ra \M_{d_A}$ satisfy
\[
W^{\otimes n}_{p_1}\otimes (\vartheta_{d_A}\circ W_{p_2})^{\otimes m}\in \mathcal{C}.
\]
\label{thm:ReducMappCon}
\end{thm} 

\begin{proof}
Note that for any $N,M\in \N$ with $M\leq N$ the maps $T:\M^{\otimes M}_{d_A}\ra\M^{\otimes N}_{d_A}$ given by $T(X)=X\otimes \one^{\otimes N-M}_{d_A}$ and the partial trace $S:\M^{\otimes N}_{d_B}\ra\M^{\otimes M}_{d_B}$ given by $S = \ident^{\otimes M}_{d_B}\otimes \text{Tr}^{\otimes N-M}$ are completely positive. Whenever $P^{\otimes N}\in \mathcal{C}$ is non-zero, and using the properties of mapping cones from Definition \ref{defn:MappingCone} we have 
\[
P^{\otimes M} = \frac{1}{x} S\circ P^{\otimes N}\circ T\in \mathcal{C}
\]
where $x = (\text{Tr}\lb P(\one_{d_A})\rb)^{N-M}>0$. Therefore, we can assume without loss of generality that $N = n+m$, and will do so in the following. 

By assumption $P$ and $\vartheta_{d_B}\circ P$ are not completely positive. Therefore, we can apply Theorem \ref{thm:SymmP} for both $p_1,p_2$ satisfying \eqref{equ:p1p2RangesFirst} to find completely positive maps $S_1,S_2:\M_{d_B}\ra\M_{d_A}$ such that 
\begin{equation}
\int_{\mathcal{U}_{d_A}} \text{Ad}_U\circ \vartheta_{d_A}\circ S_1\circ\vartheta_{d_B}\circ  P\circ \text{Ad}_{U^T} \text{d}U = W_{p_1},
\label{equ:Twirl1}
\end{equation}
and
\begin{equation}
\int_{\mathcal{U}_{d_A}} \text{Ad}_{\overline{U}}\circ S_2\circ P\circ \text{Ad}_{U^T} \text{d}U = \vartheta_{d_A}\circ W_{p_2},
\label{equ:Twirl2}
\end{equation}
where we used the identity $\text{Ad}_U\circ \vartheta_{d_A} = \vartheta_{d_A}\circ \text{Ad}_{\overline{U}}$ for any $U\in\mathcal{U}_{d_A}$. In the following we set $\tilde{S}_1 = \vartheta_{d_A}\circ S_1\circ\vartheta_{d_B}$ and note that this map is completely positive. 

Assume now that $P^{\otimes N}:\M^{\otimes N}_{d_A}\ra\M^{\otimes N}_{d_B}$ satisfies $P^{\otimes N}\in \mathcal{C}$. Then, using \eqref{equ:Twirl1} and \eqref{equ:Twirl2} we obtain
\begin{align*}
&W^{\otimes n}_{p_1}\otimes (\vartheta_{d_2}\circ W_{p_2})^{\otimes m}  \\
&= \int_{\substack{U_1,\ldots ,U_n\in\mathcal{U}_{d_A}\\V_1,\ldots ,V_m\in\mathcal{U}_{d_A} }}\bigotimes^{n}_{i=1}\lb\text{Ad}_{U_i}\circ \tilde{S}_1\circ P\circ \text{Ad}_{U^T_i}\rb \otimes \bigotimes^{m}_{j=1}\lb\text{Ad}_{\overline{V_j}}\circ S_2\circ P\circ \text{Ad}_{V^T_j}\rb \text{d}U_1\cdots \text{d}U_n\text{d}V_1\cdots \text{d}V_m\\
&= \int_{\substack{U_1,\ldots ,U_n\in\mathcal{U}_{d_A}\\V_1,\ldots ,V_m\in\mathcal{U}_{d_A} }} K^{U_1,\ldots , U_n}_{V_1,\ldots ,V_m}\circ  P^{\otimes N}\circ L^{U_1,\ldots , U_n}_{V_1,\ldots ,V_m}\text{d}U_1\cdots \text{d}U_n\text{d}V_1\cdots \text{d}V_m \in \mathcal{C}.
\end{align*} 
Here we introduced completely positive maps 
\[
K^{U_1,\ldots , U_n}_{V_1,\ldots ,V_m} = \bigotimes^{n}_{i=1}\lb\text{Ad}_{U_i}\circ \tilde{S}_1\rb \otimes \bigotimes^{m}_{j=1}\lb\text{Ad}_{\overline{V_j}}\circ S_2\rb,
\] 
and 
\[
L^{U_1,\ldots , U_n}_{V_1,\ldots ,V_m} = \bigotimes^{n}_{i=1}\lb\text{Ad}_{U^T_i}\rb \otimes \bigotimes^{m}_{j=1}\lb\text{Ad}_{V^T_j}\rb,
\] 
and used convexity of $\mathcal{C}_{d^N_A,d^N_A}$ and the second property of Definition \ref{defn:MappingCone} of $\mathcal{C}$. This finishes the proof. 
\end{proof}

Note that the family $\mathcal{D}=\lset \mathcal{D}_{d_A,d_B}\rset_{d_A,d_B\in\N}$ where $\mathcal{D}_{d_A,d_B}$ denotes the set of decomposable maps $P:\M_{d_A}\ra\M_{d_B}$ is a mapping cone according to Definition \ref{defn:MappingCone}. Therefore, we have the following corollary: 

\begin{cor}
Consider a positive map $P:\M_{d_A}\ra\M_{d_B}$ that is neither completely positive nor completely copositive. If for some $N\in\N$ the tensor power $P^{\otimes N}$ is decomposable, then for any $n,m\in\N$ with $n+m\leq N$, and any $p_1,p_2$ with 
\begin{equation}
\frac{1}{2}\lb 1 + \mu\lb \vartheta_{d_B}\circ P\rb\rb < p_1 < \frac{1}{2}\hspace*{0.5cm}\text{ and }\hspace*{0.5cm} \frac{1}{2}\lb 1 + \mu\lb P\rb\rb < p_2 < \frac{1}{2}
\label{equ:p1p2Ranges}
\end{equation}
the map 
\[
W^{\otimes n}_{p_1}\otimes (\vartheta_{d_A}\circ W_{p_2})^{\otimes m}:\M^{\otimes N}_{d_A}\ra\M^{\otimes N}_{d_A}
\]
is decomposable.
\label{cor:ReducToWerner}
\end{cor}

In the next section we will completely characterize the decomposability of the linear maps $W^{\otimes n}_{p_1}\otimes (\vartheta_{d}\circ W_{p_2})^{\otimes m}$ for $n,m\in \N$ via a linear program. With Corollary \ref{cor:ReducToWerner} this yields criteria for when general positive maps have decomposable tensor powers.   

\section{Mixed tensor powers of Werner-maps}
\label{sec:MixedTensroHW}

To study when the mixed tensor powers $W^{\otimes n}_{p_1}\otimes (\vartheta_{d}\circ W_{p_2})^{\otimes m}$ are decomposable we will first introduce a family of symmetric quantum states with positive partial transpose. In the case $n=m=1$ this family has been studied and characterized already in~\cite{vollbrecht2002activating}. Using this family we can find witnesses for non-decomposability (in the sense of Theorem \ref{thm:DecWitness}), that can be found using a linear program.  For general $n,m\in\N$ we will consider a specific example from our family of symmetric quantum states that will yield analytic bounds on the decomposability of the mixed tensor powers of Werner maps. To our knowledge the results obtained here for $(n,m)\neq (1,1)$ are new, but it should be noted that similar families of symmetric quantum states have been studied before in different context (see for example \cite{chruscinski2006multipartite,chruscinski2007quantum,audenaert2001asymptotic}).

\subsection{Symmetric states with positive partial transpose}
\label{subsec:SymmStates}

Throughout this section the dimensions $d_A,d_B\in\N$ will denote the same number $d=d_A=d_B$, and the labels $'A'$ and $'B'$ indicate a particular bipartition of a multipartite state. 

For $n,m\in\N_0$ and a real rectangular matrix $Q\in \M_{n+1,m+1}(\R)$ consider the matrix 
\begin{equation}
H_Q\in (\M_{d_A}\otimes \M_{d_B})^{\otimes n}\otimes (\M_{d_A}\otimes \M_{d_B})^{\otimes m}
\label{equ:HQParts}
\end{equation} 
given by 
\begin{equation}
H_Q = \sum^{n+1}_{k=1}\sum^{m+1}_{l=1} Q_{kl} F(k,l)
\label{equ:HQ}
\end{equation}
with 
\begin{equation}
F(k,l) = \frac{1}{n!m!}\sum_{\substack{i_1,\ldots ,i_n\in \lset 0,1\rset\\ i_1+\cdots + i_n=k-1}}~\sum_{\substack{j_1,\ldots ,j_m\in \lset 0,1\rset\\ j_1+\cdots + j_m=l-1}}P_{i_1}\otimes \cdots \otimes P_{i_n}\otimes P^\Gamma_{j_1}\otimes \cdots \otimes P^{\Gamma}_{j_m},
\label{equ:Skl}
\end{equation}
where $P_0 = \frac{P_\text{sym}}{d_{\text{sym}}}$ and $P_1 = \frac{P_\text{asym}}{d_{\text{asym}}}$. In the following we will always consider the matrix $H_Q$ as a bipartite matrix with respect to the bipartition into $A$ and $B$ systems (see the labeling of dimensions in \eqref{equ:HQParts}). 

We will now characterize the set of parameters $Q\in \M_{n+1,m+1}(\R)$ for which the matrix $H_Q$ is positive:

\begin{thm}[Positivity of $H_Q$]
For $m,d\in \N$ let $V^{m}_d\in\M_{m+1}$ be the matrix with the entries
\begin{equation}
(V^{m}_d)_{ab} = (-1)^{a-1}\sum^{\min(a-1,b-1)}_{t=\max(0,a+b-m-2)} \binom{a-1}{t}\binom{m-a+1}{b-t-1}\lb -\frac{d+1}{d-1}\rb^{t},
\label{equ:V}
\end{equation}
for $a,b\in\lset 1,\ldots, m+1\rset$. The matrix $H_Q\in (\M_{d_A}\otimes \M_{d_B})^{\otimes n}\otimes (\M_{d_A}\otimes \M_{d_B})^{\otimes m}$ with $d=d_A=d_B$ as defined in \eqref{equ:HQ} with $Q\in \M_{n+1,m+1}(\R)$ is positive iff the product $QV^{m}_d$ is entrywise positive.
\label{thm:HQPos}
\end{thm}

\begin{proof}
To prove the theorem we will first diagonalize the Hermitian matrix $H_Q\in (\M_{d_A}\otimes \M_{d_B})^{\otimes n}\otimes (\M_{d_A}\otimes \M_{d_B})^{\otimes m}$ for $Q\in \M_{n+1,m+1}(\R)$. Therefore, consider the vectors
\begin{equation}
\ket{\Psi_{x_1,\ldots,x_n,y_1,\ldots,y_m}} = \bigotimes^n_{i=1}\ket{\psi_{x_i}} \otimes \bigotimes^m_{j=1}\ket{\phi_{y_j}} 
\label{eq:EvecH}
\end{equation}
where each $x_i,y_j\in\lset 1,\ldots ,d^2\rset$ for $i\in\lset 1,\ldots ,n\rset$ and $j\in\lset 1,\ldots ,m\rset$. Here $(\ket{\psi_{x}})^{d^2}_{x=1}$ forms an orthonormal basis of $\C^d\otimes \C^d$ made from symmetric and antisymmetric vectors (say $x\in S = \lset 1,\ldots ,\frac{1}{2}d(d-1)\rset$ corresponding to symmetric and $x\in A = \lset \frac{1}{2}d(d-1)+1, \ldots ,d^2\rset$ corresponding to antisymmetric vectors), and $(\ket{\phi_{y}})^{d^2}_{y=1}$ forms an orthonormal basis of $\C^d\otimes \C^d$ made from the maximally entangled state $\frac{1}{\sqrt{d}}\ket{\Omega}$ and vectors orthogonal to this state (say $y=1$ corresponding to the maximally entangled state). Note that these basis vectors satisfy 
\[
P_{0}\ket{\psi_x} = \begin{cases} \frac{1}{d_{\text{sym}}}\ket{\psi_x}&,\text{ if }x\in S \\ 0&,\text{ otherwise,}\end{cases}\hspace{0.8cm} P_{1}\ket{\psi_x} = \begin{cases} \frac{1}{d_{\text{asym}}}\ket{\psi_x}&,\text{ if }x\in A \\ 0&,\text{ otherwise,}\end{cases}
\]
and 
\[
P^\Gamma_{0}\ket{\phi_y} = \begin{cases} \frac{1}{d}\ket{\phi_y}&,\text{ if }y=1 \\ \frac{1}{2d_{\text{sym}}}\ket{\phi_y}&,\text{ otherwise,}\end{cases}\hspace{1cm} P^\Gamma_{1}\ket{\phi_y} = \begin{cases} -\frac{1}{d}\ket{\phi_y}&,\text{ if }y=1 \\ \frac{1}{2d_{\text{asym}}}\ket{\phi_y}&,\text{ otherwise,}\end{cases}
\]
for $P_0 = P_\text{sym}/d_{\text{sym}}$ and $P_1 = P_\text{asym}/d_{\text{asym}}$. Therefore, it is easy to see that for any $x_1,\ldots , x_n\in \lset 1,\ldots ,d^2\rset$ and $y_1,\ldots , y_n\in \lset 1,\ldots ,d^2\rset$ the vector $\ket{\Psi_{x_1,\ldots,x_n,y_1,\ldots,y_m}}$ defined in \eqref{eq:EvecH} is an eigenvector of $H_Q$ as defined in \eqref{equ:HQ} for $Q\in \M_{n+1,m+1}(\R)$.

To obtain a criterion for the positivity of $H_Q$ we first compute the eigenvalues of $F(k,l)$ as defined in \eqref{equ:Skl}. Fix $k\in \lset 1,\ldots, n+1\rset$ and $l\in \lset 1,\ldots , m+1\rset$, and $x_1,\ldots , x_n,y_1,\ldots ,y_m\in \lset 1,\ldots ,d^2\rset$. Then, we have 
\begin{align*}
&\bra{\Psi_{x_1,\ldots,x_n,y_1,\ldots,y_m}}F(k,l)\ket{\Psi_{x_1,\ldots,x_n,y_1,\ldots,y_m}} \\
&\hspace{0.5cm} = \begin{cases} \displaystyle\frac{1}{\binom{n}{k-1} m! d^{n-k+1}_\text{sym}d^{k-1}_\text{asym}}\sum_{\substack{j_1,\ldots ,j_m\in \lset 0,1\rset\\ j_1+\cdots + j_m=l-1}} \,\prod^m_{i=1}\bra{\phi_{y_i}}P^\Gamma_{j_i}\ket{\phi_{y_i}}&,\text{ if }k= |\lset i~:~x_i\in A\rset|+1 ,\\ 0&,\text{otherwise.}\end{cases}
\end{align*}
Let $w = |\lset i~:~y_i\neq 1\rset|$ and for fixed $j_1,\ldots ,j_m\in \lset 0,1\rset$ satisfying $l=j_1+\cdots j_m + 1$ denote $t = |\lset i~:~j_i=1\text{ and }y_i\neq 1\rset|$. Now note that by the above
\begin{align*}
\prod^m_{i=1}\bra{\phi_{y_i}}P^\Gamma_{j_i}\ket{\phi_{y_i}} &= \lb -\frac{1}{d}\rb^{l-t-1}\lb \frac{1}{2d_{\text{asym}}}\rb^{t}\lb \frac{1}{d}\rb^{m-l-w+t+1}\lb \frac{1}{2d_{\text{sym}}}\rb^{w-t} \\
&=\frac{(-1)^{l-t-1}}{d^m (d+1)^w}\lb\frac{d+1}{d-1}\rb^t.
\end{align*}
Since for fixed parameters $w,l,d$ and $m$ the value of the product in the previous equation only depends $t\in\lset\max(0,l+w-m-1),\ldots , \min(l-1,w)\rset$ we find that
\begin{align*}
\sum_{\substack{j_1,\ldots ,j_m\in \lset 0,1\rset\\ j_1+\cdots + j_m=l-1}} \,&\prod^m_{i=1}\bra{\phi_{y_i}}P^\Gamma_{j_i}\ket{\phi_{y_i}} \\
&= \frac{w!(m-w)!}{d^m(d+1)^w}(-1)^{l-1}\sum^{\min(l-1,w)}_{t=\max(0,l+w-m-1)} \binom{l-1}{t}\binom{m-l+1}{w-t}\lb -\frac{d+1}{d-1}\rb^{t}.
\end{align*}
Here we have used that there are $\binom{l-1}{t}\binom{m-l+1}{w-t}w!(m-w)!$ many possibilities to select $t$ of the $P^\Gamma_{1}$ and $w-t$ of the $P^\Gamma_{0}$ to be paired with $\ket{\phi_{y_i}}$ for $y_i\neq 1$ (thereby fixing which of the $P^\Gamma_{1}$ and $P^\Gamma_{0}$ are paired with $\ket{\phi_{y_i}}$ for $y_i = 1$) and counting all possible permutations corresponding to the same $t$ separately. 

Now it is easy to see that the eigenvalues of $H_Q$ for $Q\in \M_{n+1,m+1}(\R)$ are given by 
\[
\bra{\Psi_{x_1,\ldots,x_n,y_1,\ldots,y_m}}H_Q\ket{\Psi_{x_1,\ldots,x_n,y_1,\ldots,y_m}} = \frac{1}{\binom{n}{k-1} \binom{m}{w} d^{n-k+1}_\text{sym}d^{k-1}_\text{asym}d^m(d+1)^w} (QV^{m}_d)_{k(w+1)}
\]
for $k-1= |\lset i~:~x_i\in A\rset|\in \lset 0,\ldots ,n\rset$ and $w = |\lset i~:~y_i\neq 1\rset|\in \lset 0,\ldots ,m\rset$. These numbers are all positive iff the product $QV^{m}_d$ is entrywise positive.

\end{proof} 

Note that for $Q\in \M_{n+1,m+1}(\R)$ the matrices $H^\Gamma_Q$ and $H_{Q^T}$ are unitarily equivalent by exchanging the first $n$ pairs of $\C^d$ with the final $m$ pairs of $\C^d$. Therefore, we immediately get the following corollary.

\begin{cor}[Positive partial transpose of $H_Q$]
For $n,d\in \N$ let $V^{n}_d\in\M_{n+1}$ be defined as in $\eqref{equ:V}$. The matrix $H_Q\in (\M_{d_A}\otimes \M_{d_B})^{\otimes n}\otimes (\M_{d_A}\otimes \M_{d_B})^{\otimes m}$ with $d=d_A=d_B$ as defined in \eqref{equ:HQ} with $Q\in \M_{n+1,m+1}(\R)$ has positive partial transpose (with respect to the bipartition into $'A'$ and $'B'$ systems) iff the product $Q^TV^{n}_d$ is entrywise positive.
\label{cor:HQPPT}

\end{cor}

\subsection{Linear program for decomposability}
\label{sec:DecompHW}

Before we can state our main result, we will have to introduce some notation. For $n\in\N$ and $p\in\lbr 0,1\rbr$ let $\ket{v^n_p}\in \R^{n+1}$ denote the vector with entries
\begin{equation}
(\ket{v^n_p})_k = \lb\frac{d+1}{d-1}\rb^{k-1} (1-p)^{k-1} p^{n-k+1}
\label{equ:vecV}
\end{equation}
for $k\in\lset 1,\ldots ,n+1\rset$. The following theorem characterizes the decomposability of mixed tensor powers of Werner maps in terms of a linear program. 

\begin{thm}[Linear program for decomposability]
For $d,n,m\in\N$ let $W_{p_1},W_{p_2}:\M_d\ra\M_d$ denote Werner maps with parameters $p_1,p_2\in\lbr 0,1\rbr$ (see Section \ref{sec:WernerStatesAndMaps}). Then, the mixed tensor power $W^{\otimes n}_{p_1}\otimes (\vartheta_{d}\circ W_{p_2})^{\otimes m}$ is decomposable iff the value of the linear program
\begin{align*}
\text{minimize}\hspace{0.3cm}&\bra{v^n_{p_1}}Q\ket{v^m_{p_2}} \\
\text{subject to}\hspace{0.3cm}&Q\in \M_{n+1,m+1}(\R) \\
&\bra{i}QV^{m}_d\ket{j}\geq 0, \\
&\bra{j}Q^TV^{n}_d\ket{i}\geq 0, \text{ for all }i\in\lset 1,\ldots , n+1\rset\text{ and }j\in\lset 1,\ldots , m+1\rset\\
& \sum^{n+1}_{k=1}\sum^{m+1}_{l=1} Q_{kl} = 1
\end{align*}  
is positive. Here $V^{n}_d$ denotes the matrix introduced in \eqref{equ:V} and $\ket{v^n_{p}}$ the vector from \eqref{equ:vecV}.

\label{thm:LinProgrForDec}
\end{thm}

\begin{proof}
Let $d,n,m\in\N$ be fixed, and note that the Choi matrix of the map $W^{\otimes n}_{p_1}\otimes (\vartheta_{d}\circ W_{p_2})^{\otimes m}$ is unitarily equivalent to $\rho_W(p_1)^{\otimes n}\otimes (\rho_W(p_2)^\Gamma)^{\otimes m}$ with $\rho_W(p)\in(\M_d\otimes \M_d)^+$ denoting the Werner states (see \eqref{equ:WernerStates}). By Theorem \ref{thm:DecWitness} the map $W^{\otimes n}_{p_1}\otimes (\vartheta_{d}\circ W_{p_2})^{\otimes m}$ is decomposable iff
\begin{equation}
\min_{H}\text{Tr}\lbr \lb\rho_W(p_1)^{\otimes n}\otimes (\rho_W(p_2)^\Gamma)^{\otimes m}\rb H\rbr\geq 0,
\label{equ:decMin}
\end{equation}  
where the minimization is over $H\in (\M_{d}\otimes \M_{d})^{\otimes n+m}$ with positive partial transpose and trace equal to $1$. 

We will now impose symmetries on the minimization problem \eqref{equ:decMin} to obtain the final linear program in the theorem. Note first, that the twirl map $T_{UU}:\M_{d}\otimes \M_d\ra \M_{d}\otimes \M_d$ defined in \eqref{equ:Twirl} is selfadjoint (with respect to the Hilbert-Schmidt inner product) and leaves $\rho_W(p)$ invariant. Therefore, we have for any $H\in (\M_{d}\otimes \M_{d})^{\otimes n+m}$ that 
\begin{align}
\nonumber\text{Tr}&\lbr \lb\rho_W(p_1)^{\otimes n}\otimes (\rho_W(p_2)^\Gamma)^{\otimes m}\rb H\rbr\\
&= \nonumber\text{Tr}\lbr \lb T_{UU}\lb\rho_W(p_1)\rb^{\otimes n}\otimes (T_{UU}\lb\rho_W(p_2)\rb^\Gamma)^{\otimes m}\rb H\rbr \\
&= \text{Tr}\lbr \lb\rho_W(p_1)^{\otimes n}\otimes \lb\rho_W(p_2)^\Gamma\rb^{\otimes m}\rb \lb T^{\otimes n}_{UU}\otimes \lb(\ident_d\otimes \vartheta_d)\circ T_{UU}\circ (\ident_d\otimes \vartheta_d)\rb^{\otimes m} \rb (H)\rbr.
\label{equ:IntermediateStepMin}
\end{align}
By \eqref{equ:Twirl} and using that the transposition is selfadjoint (with respect to the Hilbert-Schmidt inner product) we have 
\[
\lb(\ident_d\otimes \vartheta_d)\circ T_{UU}\circ (\ident_d\otimes \vartheta_d)\rb(X) = \text{Tr}\lbr X P^\Gamma_{\text{sym}}\rbr\frac{P^\Gamma_{\text{sym}}}{d_{\text{sym}}} + \text{Tr}\lbr X P^\Gamma_{\text{asym}}\rbr\frac{P^\Gamma_{\text{asym}}}{d_{\text{asym}}}.
\]
Using \eqref{equ:Twirl} again we obtain 
\begin{align}
\nonumber( T^{\otimes n}_{UU}&\otimes \lb(\ident_d\otimes \vartheta_d)\circ T_{UU}\circ (\ident_d\otimes \vartheta_d)\rb^{\otimes m} ) (H) \\
&= \sum_{i_1,\ldots ,i_n\in \lset 0,1\rset}\sum_{j_1,\ldots ,j_m\in \lset 0,1\rset} q^{i_1\ldots i_n}_{j_1\ldots j_m} P_{i_1}\otimes \cdots \otimes P_{i_n}\otimes P^\Gamma_{j_1}\otimes \cdots \otimes P^{\Gamma}_{j_m}
\label{equ:NewH}
\end{align}
with $P_0 = \frac{P_\text{sym}}{d_{\text{sym}}}$ and $P_1 = \frac{P_\text{asym}}{d_{\text{asym}}}$ and 
\[
q^{i_1\ldots i_n}_{j_1\ldots j_m} = d^{n+m-i-j}_{\text{sym}}d^{i+j}_{\text{asym}}\text{Tr}\lbr \lb P_{i_1}\otimes \cdots \otimes P_{i_n}\otimes P^\Gamma_{j_1}\otimes \cdots \otimes P^\Gamma_{j_m}\rb H \rbr
\]
with $i= \sum^n_{k=1} i_k$ and $j= \sum^m_{l=1} j_l$. Note that since $P^\Gamma_{1}$ is \emph{not} positive some of the coefficients $q^{i_1\ldots i_n}_{j_1\ldots j_m}$ can take negative values. By \eqref{equ:IntermediateStepMin} we can simplify the minimization in \eqref{equ:decMin} by restricting to positive $H$ with positive partial transpose of the form \eqref{equ:NewH} and with trace equal to $1$. 

For any permutation $\sigma\in S_n$ we denote by $U^{(n)}_\sigma:(\C^d\otimes \C^d)^{\otimes n}$ the unitary acting as
\[
U^{(n)}_\sigma (\ket{v_1}\otimes \cdots \otimes\ket{v_n}) = \ket{v_{\sigma(1)}}\otimes \cdots \otimes\ket{v_{\sigma(n)}}
\]
for any $\ket{v_1},\ldots , \ket{v_n}\in (\C^d\otimes \C^d)$. Then, we can define the symmetrization map $\mathcal{S}_n:(\M_d\otimes \M_d)^{\otimes n}\ra(\M_d\otimes \M_d)^{\otimes n}$ by 
\[
\mathcal{S}_n(X) = \frac{1}{n!}\sum_{\sigma\in S_n} U^{(n)}_\sigma X (U^{(n)}_\sigma)^\dagger,
\] 
for any $X\in(\M_d\otimes \M_d)^{\otimes n}$. With this map we have
\begin{align}
\nonumber\text{Tr}\lbr \lb\rho_W(p_1)^{\otimes n}\otimes (\rho_W(p_2)^\Gamma)^{\otimes m}\rb H\rbr &= \text{Tr}\lbr \lb\mathcal{S}_n\lb\rho_W(p_1)^{\otimes n}\rb\otimes \mathcal{S}_m\lb(\rho_W(p_2)^\Gamma)^{\otimes m}\rb\rb H\rbr \\
&= \text{Tr}\lbr \lb\rho_W(p_1)^{\otimes n}\otimes (\rho_W(p_2)^\Gamma)^{\otimes m}\rb \lb\mathcal{S}_n\otimes \mathcal{S}_m\rb\lb H\rb\rbr,
\label{equ:PermSymm}
\end{align}
for any $H\in (\M_{d}\otimes \M_{d})^{\otimes n+m}$. Now note that for $H$ of the form \eqref{equ:NewH} with coefficients $q^{i_1\ldots i_n}_{j_1\ldots j_m}\in \R$ we obtain 
\[
\lb\mathcal{S}_n\otimes \mathcal{S}_m\rb\lb H\rb = \sum^{n+1}_{k=1}\sum^{m+1}_{l=1} Q_{kl} F(k,l)
\]
with 
\[
F(k,l) = \frac{1}{n!m!}\sum_{\substack{i_1,\ldots ,i_n\in \lset 0,1\rset\\ i_1+\cdots + i_n=k-1}}\sum_{\substack{j_1,\ldots ,j_m\in \lset 0,1\rset\\ j_1+\cdots + j_m=l-1}}P_{i_1}\otimes \cdots \otimes P_{i_n}\otimes P^\Gamma_{j_1}\otimes \cdots \otimes P^{\Gamma}_{j_m},
\]
and
\[
Q_{kl} = \sum_{\substack{i_1,\ldots ,i_n\in \lset 0,1\rset\\ i_1+\cdots + i_n=k-1}}\sum_{\substack{j_1,\ldots ,j_m\in \lset 0,1\rset\\ j_1+\cdots + j_m=l-1}} q^{i_1\ldots i_n}_{j_1\ldots j_m}.
\]
We can therefore simplify the minimization problem \eqref{equ:decMin} further, by restricting $H$ to the family of states $H_Q$ with positive partial transpose as in \eqref{equ:HQ} parametrized by a real matrix $Q\in \M_{n+1,m+1}(\R)$.

Given $Q\in \M_{n+1,m+1}(\R)$ we can use the definition \eqref{equ:WernerStates} of the Werner states to compute  
\begin{align*}
\text{Tr}&\lbr \lb\rho_W(p_1)^{\otimes n}\otimes (\rho_W(p_2)^\Gamma)^{\otimes m}\rb H_Q\rbr \\
&= \sum^{n-1}_{k=1}\sum^{m-1}_{l=1} Q_{kl} \text{Tr}\lbr \lb\rho_W(p_1)^{\otimes n}\otimes (\rho_W(p_2)^\Gamma)^{\otimes m}\rb F(k,l)\rbr \\
& = \sum^{n-1}_{k=1}\sum^{m-1}_{l=1} Q_{kl} \frac{1}{d_{sym}^{n+m-k-l+2}d_{asym}^{k+l-2}} (1-p_1)^{k-1}(1-p_2)^{l-1}p_1^{n-k+1}p_2^{m-l+1}\\
& = \frac{2^{n+m}}{d^{n+m}(d+1)^{n+m}}\bra{v^n_{p_1}}Q\ket{v^m_{p_2}},
\end{align*}
with $\ket{v^n_{p}}$ as defined in \eqref{equ:vecV}.
Now the minimization \eqref{equ:decMin} to
\begin{equation}
\min_{H}\text{Tr}\lbr \lb\rho_W(p_1)^{\otimes n}\otimes (\rho_W(p_2)^\Gamma)^{\otimes m}\rb H\rbr
 = \frac{2^{n+m}}{d^{n+m}(d+1)^{n+m}}\min_{Q} \bra{v^n_{p_1}}Q\ket{v^m_{p_2}},
 \label{equ:FinalOpt}
\end{equation}
where the final minimization is over all matrices $Q\in \M_{n+1,m+1}(\R)$ corresponding to positive matrices $H_Q$ via \eqref{equ:HQ} with positive partial transpose and satisfying $\sum^{n}_{k=1}\sum^m_{l=1} Q_{kl}=1$. Using Theorem \ref{thm:HQPos} and Corollary \ref{cor:HQPPT} these are the matrices $Q\in \M_{n+1,m+1}(\R)$ for which $QV^{m}_d$ and $Q^TV^{n}_d$ are entrywise positive, with $V^{n}_d$ (and $V^{m}_d$) as defined in \eqref{equ:V}. With these conditions and omitting the positive factor in \eqref{equ:FinalOpt} (since we are only interested in positivity of the minimization problem) we obtain the linear program stated in the theorem.

\end{proof}

It should be noted, that the linear program from the previous theorem should be seen mainly as a theoretical result. When $n,m\in\N$ are small enough it can be used to numerically determine the parameter regions where the maps $W^{\otimes n}_{p_1}\otimes (\vartheta_{d}\circ W_{p_2})^{\otimes m}$ are decomposable. However, for large $n,m\in\N$ one has to be careful since some entries of $V^{n}_d$ (see \eqref{equ:V}) get very large in magnitude, and the alternating signs in $V^{n}_d$ can cause numerical instabilities.

\subsection{Special case: $n=m=1$}
\label{sec:SpecialCase}

For $n=m=1$ the linear program from Theorem \ref{thm:LinProgrForDec} has a simple form, and the set of feasible points has been characterized already in \cite{vollbrecht2002activating}. Using this we can derive an analytic condition for decomposability of the map $W_{p_1}\otimes (\vartheta_{d}\circ W_{p_2})$. Specifically, we will show

\begin{thm}
For $d\in\N$ and $p_1,p_2\in \lbr 0,1\rbr$ consider the Werner maps $W_{p_1},W_{p_2}:\M_d\ra\M_d$ as defined in Section \ref{sec:WernerStatesAndMaps}. Then, the map $W_{p_1} \otimes (\vartheta_d\circ W_{p_2})$ is decomposable iff  
\begin{equation}
 \frac{d-1}{d+1}p_1p_2 + (1-p_1)p_2 + p_1(1-p_2) - (1-p_1)(1-p_2) \geq  0.
\label{equ:n1m1}
\end{equation} 
\label{thm:n1m1}
\end{thm} 

For $p_1=p_2$ the bound \eqref{equ:n1m1} can be further simplified:

\begin{cor}
For $d\in\N$ and $p\in \lbr 0,1\rbr$ consider the Werner map $W_{p}:\M_d\ra\M_d$ as defined in Section \ref{sec:WernerStatesAndMaps}. Then, the map $W_{p} \otimes (\vartheta_d\circ W_{p})$ is decomposable iff  
\begin{equation}
p \geq \frac{1}{2+\sqrt{\frac{2d}{d+1}}}.
\label{equ:n1m1Equalp}
\end{equation}
\label{cor:n1m1Equalp}
\end{cor} 

To prove Theorem \ref{thm:n1m1} we will analyze the convex set 
\begin{equation}
\mathcal{C}_d :=\left\lbrace Q\in \M_{2,2}(\R)~:~QV^{1}_d \text{ and }Q^TV^{1}_d\text{ entrywise positive and }\sum^2_{i,j=1}Q_{ij}=1\right\rbrace 
\label{equ:Cset}
\end{equation}
of feasible points for the linear program from Theorem \ref{thm:n1m1} for some $d\in\N$ in the case $n=m=1$. Here, by \eqref{equ:V} we have 
\begin{equation}
V^1_d = \begin{pmatrix} 1 & 1 \\ -1 & \alpha_d \end{pmatrix}
\label{equ:Vmat2}
\end{equation}
for $\alpha_d = (d+1)/(d-1)$. We will need the following lemma originally shown in \cite{vollbrecht2002activating} using a different parametrization:

\begin{lem}[Extreme points\cite{vollbrecht2002activating}]
For $d\in \N$ the extreme points of $\mathcal{C}_d$ are given by 
\begin{align*}
Q_1 &= \frac{1}{2(d+2)}\begin{pmatrix} d+1 &  d+1\\  d+1 & -(d-1)\end{pmatrix}, \\
Q_2 &= \begin{pmatrix} 1 & 0 \\ 0 & 0 \end{pmatrix}, \\
Q_3 &= \frac{1}{4}\begin{pmatrix} 1 & 1 \\ 1 & 1 \end{pmatrix}, \\
Q_4 &= \frac{1}{2}\begin{pmatrix} 1 & 0 \\ 1 & 0 \end{pmatrix}, \\
Q_5 &= \frac{1}{2}\begin{pmatrix} 1 & 1 \\ 0 & 0 \end{pmatrix}.
\end{align*} 
\label{lem:ExtremePoints}
\end{lem}

\begin{proof}
For completeness we present a proof of this result in Appendix \ref{sec:Appendix}. 
\end{proof}

Now it is easy to prove Theorem \ref{thm:n1m1}.

\begin{proof}[Proof of Theorem \ref{thm:n1m1}]
For $d\in \N$ and $p_1,p_2\in \lbr 0,1\rbr$ consider the Werner maps $W_{p_1},W_{p_2}:\M_d\ra\M_d$ as defined in Section \ref{sec:WernerStatesAndMaps}. By Theorem \ref{thm:LinProgrForDec} the map $W_{p_1}\otimes (\vartheta_d\circ W_{p_2})$ is decomposable iff the value of 
\begin{equation}
\min_{Q\in \mathcal{C}_d}\bra{v^1_{p_1}}Q\ket{v^1_{p_2}} \geq 0,
\label{equ:minn1m1}
\end{equation}   
with 
\[
\ket{v^1_{p}} = \begin{pmatrix} p \\ \alpha_d (1-p)\end{pmatrix}
\]
as in \eqref{equ:vecV} and $\mathcal{C}_d$ as in \eqref{equ:Cset}. By Lemma \ref{lem:ExtremePoints} the minimum in \eqref{equ:minn1m1} is attained in one of the points $Q_1,\ldots , Q_5$. Now, observe that for any $p_1,p_2\in \lbr 0,1\rbr$ and $d\in\N$ we have $\bra{v^1_{p_1}}Q_i\ket{v^1_{p_2}}\geq 0$ for $i\in\lset 2,3,4,5\rset$. Thus, the map $W_{p_1}\otimes (\vartheta_d\circ W_{p_2})$ is decomposable iff 
\[
\bra{v^1_{p_1}}Q_1\ket{v^1_{p_2}} = \frac{(d+1)^2}{2(d+2)(d-1)}\lb \frac{d-1}{d+1}p_1p_2 + (1-p_1)p_2 + p_1(1-p_2) - (1-p_1)(1-p_2)\rb\geq  0.
\]
This finishes the proof.
\end{proof}

\subsection{Analytical bounds on decomposability}
\label{sec:AnalytBoundHW}
As in Section \ref{subsec:SymmStates} the dimensions $d_A,d_B\in\N$ will denote the same number $d=d_A=d_B$, and the labels $'A'$ and $'B'$ indicate a particular bipartition of a multipartite state.

To obtain an analytic bound in the case of general $n=m$ we can generalize the extreme point $Q_1$ from the previous section. Consider the matrix $Q\in \M_{n+1,n+1}(\R)$ with entries 
\begin{equation}
Q_{kl} = \begin{cases} (d-1)(d+1)^{2n} &, \text{ if }k=l=1 \\
(d-1)^{n}(d+1)^{n+1} &, \text{ if }(k,l)=(n+1,1) \text{ or }(k,l)=(1,n+1) \\
-(d-1)^{2n}(d+1) &, \text{ if }k=l=n+1. \end{cases}
\label{equ:QForAnBoun}
\end{equation}
Via \eqref{equ:HQ} this corresponds to the matrix $H_Q\in (\M_{d_A}\otimes \M_{d_B})^{\otimes n}\otimes (\M_{d_A}\otimes \M_{d_B})^{\otimes n}$ given by 
\begin{align}
\begin{split}
H_Q = \frac{2^{2n}(d+1)}{d^{2n}(n!)^2}&\bigg[ \frac{d-1}{d+1}P^{\otimes n}_{\text{sym}}\otimes (P^{\Gamma}_{\text{sym}})^{\otimes n} +\cdots \\
& \cdots+ P^{\otimes n}_{\text{sym}}\otimes (P^{\Gamma}_{\text{asym}})^{\otimes n} + P^{\otimes n}_{\text{asym}}\otimes (P^{\Gamma}_{\text{sym}})^{\otimes n} - P^{\otimes n}_{\text{asym}}\otimes (P^{\Gamma}_{\text{asym}})^{\otimes n}\bigg].
\end{split}
\label{equ:HQForAnBoun}
\end{align}
 
We start with the following lemma.

\begin{lem}
For $d_A=d_B=d$ the matrix $H_Q\in (\M_{d_A}\otimes \M_{d_B})^{\otimes n}\otimes (\M_{d_A}\otimes \M_{d_B})^{\otimes n}$ given by $\eqref{equ:HQForAnBoun}$ is positive and has positive partial transpose (with respect to the bipartition indicated by the labels $'A'$ and $'B'$). 
\label{lem:SpecificHQ}
\end{lem}

\begin{proof}
By symmetry the matrix $H_Q$ from \eqref{equ:HQForAnBoun} is positive iff it has positive partial transpose. Using Theorem \ref{thm:HQPos} the matrix $H_Q$ is positive iff the product $QV^n_d$ is entrywise positive, with $V^n_d$ as defined in \eqref{equ:V}. Note that 
\begin{equation}
(V^n_d)_{1l} = \binom{n}{l} \hspace{0.5cm}\text{ and }\hspace{0.5cm} (V^n_d)_{nl} = (-1)^{n-l+1}\binom{n}{l-1}\lb \frac{d+1}{d-1}\rb^{l-1}
\label{equ:Vvalues}
\end{equation}
for any $l\in\lset 1, \ldots ,n+1\rset$. Now we compute
\[
(QV^n_d)_{kl} = \begin{cases} \binom{n}{l-1}(d+1)^{n+l-1}\lbr (d-1)(d+1)^{n-l+1} + (1-d)^{n-l+1}(d+1)\rbr &,\text{ if } k=1\\
\binom{n}{l-1}(d+1)^{l}(d-1)^n \lbr (d+1)^{n-l+1} + (1-d)^{n-l+1}\rbr &,\text{ if } k=n+1 \\
0 &,\text{ otherwise.}
\end{cases}
\]
These entries are positive for all $k,l\in \lset 1,\ldots ,n+1\rset$.

\end{proof}

We can now prove an analytical bound on the decomposability of $W^{\otimes n}_{p_1}\otimes (\vartheta_{d}\circ W_{p_2})^{\otimes n}$. 

\begin{thm}
For $d\in\N$ and $p_1,p_2\in \lbr 0,1\rbr$ consider the Werner maps $W_{p_1},W_{p_2}:\M_d\ra\M_d$ as defined in Section \ref{sec:WernerStatesAndMaps}. Then, given $n\in\N$ whenever  
\begin{equation}
 \frac{d-1}{d+1}p_1^{n}p_2^{n} + (1-p_1)^{n}p_2^n + p_1^n(1-p_2)^{n} - (1-p_1)^{n}(1-p_2)^{n} < 0,
\label{equ:AnalyticPBoundDiffp}
\end{equation}
the map $W^{\otimes n}_{p_1} \otimes (\vartheta_d\circ W_{p_2})^{\otimes n}$ is not decomposable.  
\label{thm:AnalyticPBoundDiffp}
\end{thm}

\begin{proof}
By Lemma \ref{lem:SpecificHQ} the matrix $Q\in \M_{n+1,n+1}(\R)$ from \eqref{equ:QForAnBoun} is a feasible point of the linear program stated in Theorem \ref{thm:LinProgrForDec}. With $\ket{v^{n}_{p_1}},\ket{v^{n}_{p_2}}\in\R^{n+1}$ as defined in \eqref{equ:vecV} we find that 
\begin{align*}
\bra{v^{n}_{p_1}}Q\ket{v^{n}_{p_2}} &= \sum^{n+1}_{k=1}\sum^{n+1}_{l=1}Q_{kl} (1-p_1)^{k-1}(1-p_2)^{l-1}p_1^{(n+1)-k}p_2^{(n+1)-l}\lb\frac{d+1}{d-1}\rb^{k+l-2} \\
&= (d+1)^{2n+1}\lbr \frac{d-1}{d+1}p_1^{n}p_2^{n} + (1-p_1)^{n}p_2^n + p_1^n(1-p_2)^{n} - (1-p_1)^{n}(1-p_2)^{n}\rbr.
\end{align*}
Thus, if $p_1,p_2\in \lbr 0,1\rbr$ satisfy \eqref{equ:AnalyticPBoundDiffp}, then the minimization problem from Theorem \ref{thm:LinProgrForDec} has a negative solution, and the linear map $W^{\otimes n}_{p_1} \otimes (\vartheta_d\circ W_{p_2})^{\otimes n}$ is not decomposable.

\end{proof}

In the special case $p_1 = p_2$ we can simplify \eqref{equ:AnalyticPBoundDiffp} to obtain:

\begin{thm}
For $d\in\N$ consider the Werner map $W_p:\M_d\ra\M_d$ as defined in Section \ref{sec:WernerStatesAndMaps}. Then, given $n\in\N$ for any  
\begin{equation}
p< \frac{1}{1+\sqrt[n]{1+\sqrt{\frac{2d}{d+1}}}}<\frac{1}{2}
\label{equ:AnalyticPBound}
\end{equation}
the map $W^{\otimes n}_p \otimes (\vartheta_d\circ W_p)^{\otimes n}$ is not decomposable. 
\label{thm:AnalyticPBound}
\end{thm}

It should be noted that the bound from \eqref{equ:AnalyticPBound} is monotonically decreasing in the dimension $d$. Therefore, we can take the limit $d\ra\infty$ to obtain the bound
\[
p < \frac{1}{1+\sqrt[n]{1+\sqrt{2}}} 
\]
for which the map $W^{\otimes n}_p \otimes (\vartheta_d\circ W_p)^{\otimes n}$ is not decomposable for any $d\in \N$.

For $n\in\N$ small enough (to avoid numerical instabilities) we can use the linear program from Theorem \ref{thm:LinProgrForDec} to approximate the maximal $p_d^*(n)\in\lbr 0,1\rbr$ such that $W^{\otimes n}_p \otimes (\vartheta_d\circ W_p)^{\otimes n}$ (with local dimension $d$) is not decomposable for any $p < p_d^*(n)$. In Table \ref{table:Values} we compare these values with our analytic bound \eqref{equ:AnalyticPBound}. Note that for $n=1$ our bound gives the exact value $p_d^*(1)$ as shown in Section \ref{sec:SpecialCase}. 

\begin{table}
\begin{tabular}{l r!{\vline width 1.3pt}*{8}{c|}}
 &\hspace*{0.1cm}$n$ &  \hspace*{0.2cm}1\hspace*{0.2cm} & \hspace*{0.2cm}2\hspace*{0.2cm} & \hspace*{0.2cm}3\hspace*{0.2cm} & \hspace*{0.2cm}4\hspace*{0.2cm} & \hspace*{0.2cm}5\hspace*{0.2cm} & \hspace*{0.2cm}6\hspace*{0.2cm} & \hspace*{0.2cm}7\hspace*{0.2cm} & \hspace*{0.2cm}8\hspace*{0.2cm} \\
\Xhline{1.3pt}
$d=2$ & Numerical & 0.3169 & 0.4054 & 0.4367 & 0.4524 & 0.4619 & 0.4682 & 0.4728 & 0.4762  \\
 & Analytic bound  & 0.3169 & 0.4052 & 0.4363 & 0.4521 & 0.4616 & 0.4680 & 0.4726 & 0.4760  \\
\hline
$d=3$ & Numerical & 0.3101 & 0.4017 & 0.4340 & 0.4505 & 0.4603 & 0.4669 & 0.4716 & 0.4752  \\
& Analytic bound  & 0.3101 & 0.4013 & 0.4337 & 0.4501 & 0.4601 & 0.4667 & 0.4714 & 0.4750  \\
\hline
$d=5$ & Numerical & 0.3038 & 0.3981 & 0.4316 & 0.4486 & 0.4588 & 0.4656 & 0.4705 & 0.4742  \\
& Analytic bound  & 0.3038 & 0.3978 & 0.4313 & 0.4483 & 0.4586 & 0.4655 & 0.4704 & 0.4741  \\
\hline
$d=10$& Numerical & 0.2986 & 0.3950 & 0.4294 & 0.4469 & 0.4575 & 0.4645 & 0.4696 & 0.4734  \\
& Analytic bound   & 0.2986 & 0.3948 & 0.4293 & 0.4468 & 0.4574 & 0.4644 & 0.4695 & 0.4733  \\
\end{tabular}
\caption{Comparison between the analytic bound from Theorem \ref{thm:AnalyticPBound} and approximations obtained via the linear program from Theorem \ref{thm:LinProgrForDec} of the largest $p_d^*(n)\in\lbr 0,1\rbr$ such that $W^{\otimes n}_p \otimes (\vartheta_d\circ W_p)^{\otimes n}$ (with local dimension $d$) is not decomposable for any $p < p_d^*(n)$. All values were truncated (not rounded) to $4$ digits.}
\label{table:Values}
\end{table}

\section{Proofs of main results}
\label{sec:Proofs}

\begin{proof}[Proof of Theorem \ref{thm:QuantBound}]

Let $P:\M_{d_A}\ra\M_{d_B}$ be a non-zero $n$-tensor-stable decomposable map. Since $P$ is decomposable iff the adjoint map $P^*:\M_{d_B}\ra\M_{d_A}$ is decomposable, and $(P^{\otimes n})^* = (P^*)^{\otimes n}$, we can assume without loss of generality that $d = d_B\leq d_A$. By Lemma \ref{lem:BasicMu} the theorem holds trivially when $P$ is either completely positive or completely copositive. We can therefore assume that $P$ does not belong to either of these classes. 

Let $\mu = \max\lb \mu\lb P\rb,\mu\lb \vartheta_{d}\circ P\rb\rb<0$ and consider $1/2 > p > (1+\mu)/2$. By Corollary \ref{cor:ReducToWerner} we conclude that 
\[
W^{\otimes \left\lfloor \frac{n}{2}\right\rfloor}_p\otimes (\vartheta_d\circ W)^{\otimes \left\lfloor \frac{n}{2}\right\rfloor}_p :\M^{\otimes N}_d\ra \M^{\otimes N}_d, 
\]  
for $N=2\lfloor \frac{n}{2}\rfloor$ is decomposable. Finally, using Theorem \ref{thm:AnalyticPBound} finishes the proof.

\end{proof}

\begin{proof}[Proof of Theorem \ref{thm:NoTensStabDec}]
Let $P:\M_{d_A}\ra\M_{d_B}$ be a non-zero tensor-stable decomposable map. By Theorem \ref{thm:QuantBound} we have 
\[
\max\lb\mu\lb P\rb,\mu\lb\vartheta_{d_B}\circ P\rb\rb\geq \lb\frac{2}{1+\sqrt[\left\lfloor \frac{n}{2}\right\rfloor~]{1+\sqrt{\frac{2d}{d+1}}}} - 1\rb,
\]
for any $n\in \N$. Since the right hand side of the previous inequality converges to $0$, we can conclude that 
\[
\max\lb\mu\lb P\rb,\mu\lb\vartheta_{d_B}\circ P\rb\rb\geq 0.
\]
By Lemma \ref{lem:BasicMu} we have that $P$ is either completely positive or completely copositive. This finishes the proof.
\end{proof}

Applying some results from the theory of entanglement distillation an alternative proof of Theorem \ref{thm:NoTensStabDec} can be obtained. Before presenting this proof we will need some basic definitions: 

Given dimensions $d_A,d_B,d'_{A},d'_{B}\in\N$ we call a completely positive map $L:\M_{d_A}\otimes \M_{d_B}\ra \M_{d'_{A}}\otimes \M_{d'_{B}}$ a \emph{separable operation} if there exists an $l\in\N$ and rectangular matrices $\lset A_i\rset^l_{i=1}\subset\M_{d'_A,d_{A}}$ and $\lset B_i\rset^l_{i=1}\subset\M_{d'_B,d_{B}}$ such that 
\[
L(X) = \sum^l_{i=1} (A_i\otimes B_i)X(A_i\otimes B_i)^{\dagger}. 
\]
Similarly, we call a completely positive map $L:\M_{d_A}\otimes \M_{d_B}\ra \M_{d'_{A}}\otimes \M_{d'_{B}}$ \emph{PPT preserving} (with PPT abbreviating ``positive partial transpose'') if $L(X^\Gamma)^\Gamma\geq 0$ for any $X\in (\M_{d_A}\otimes \M_{d_B})^+$.

The following theorem has been shown in \cite{eggeling2001distillability,vollbrecht2002activating}:

\begin{thm}[Distillation via PPT-preserving completely positive maps~\cite{eggeling2001distillability,vollbrecht2002activating}]
For any quantum state $\rho\in(\M_{d_A}\otimes \M_{d_B})^+$ satisfying $\rho^\Gamma\ngeq 0$ there exists a sequence $L_{n}:\M_{d^n_A}\otimes \M_{d^n_B}\ra\M_2\otimes \M_2$ of PPT-preserving completely positive maps such that $\lim_{n\ra\infty}\|\omega_2/2 - L_n\lb \rho^{\otimes n}\rb\|_1 = 0$. Moreover, the maps $L_{n}$ can be chosen of the form 
\begin{equation}
L_{n}(X) = S_n(X\otimes \sigma^{\otimes n})
\label{equ:DistMapsPPT}
\end{equation}
for a state $\sigma\in (\M_{d_A}\otimes \M_{d_A})^+$ with positive partial transpose $\sigma^\Gamma\geq 0$, and 
\[
S_{n}:\M_{d^{2n}_A}\otimes \M_{d^{n}_Bd^{n}_A}\ra\M_2\otimes \M_2
\]
a sequence of separable operations.
\label{thm:DistPPTpres}
\end{thm}

We will need another lemma before presenting the alternative proof of Theorem \ref{thm:NoTensStabDec}:

\begin{lem}[PPT-preserving operation preserve decomposability]
Given a decomposable map $P:\M_{d_A}\ra\M_{d_B}$ and a PPT-preserving map $L:\M_{d_A}\otimes \M_{d_B}\ra \M_{d'_{A}}\otimes \M_{d'_{B}}$, the map $Q:\M_{d'_A}\ra\M_{d'_B}$ defined via the Choi matrix $C_{Q}=L\lb C_P\rb$ is decomposable as well.
\label{lem:PPTPresAndDecomp}
\end{lem}

\begin{proof}
Since $P:\M_{d_A}\ra\M_{d_B}$ is decomposable its Choi matrix can be written as
\[
C_P = X + Y^\Gamma
\]
for positive matrices $X,Y\in (\M_{d_A}\otimes \M_{d_B})^+$. Now any PPT-preserving map $L:\M_{d_A}\otimes \M_{d_B}\ra \M_{d'_{A}}\otimes \M_{d'_{B}}$ satisfies $L(X)\in (\M_{d_A}\otimes \M_{d_B})^+$ and $L(Y^\Gamma)^\Gamma\in (\M_{d_A}\otimes \M_{d_B})^+$. Therefore, we have
\[
C_Q = L\lb C_P\rb =  L(X) + \lb L(Y^\Gamma)^\Gamma\rb^\Gamma,
\]
and the linear map $Q:\M_{d'_A}\ra\M_{d'_B}$ defined by this expression is decomposable (using the Choi-Jamiolkowski isomorphism \cite{choi1975completely}). 
\end{proof}

Finally, we can present the alternative proof of our qualitative result: 

\begin{proof}[Alternative proof of Theorem \ref{thm:NoTensStabDec}]
Assume that $P:\M_{d_A}\ra\M_{d_B}$ is a non-zero tensor-stable decomposable map, that is neither completely positive nor completely copositive. Then, by Corollary \ref{cor:ReducToWerner} for any $p_1$ and $p_2$ with 
\begin{equation}
\frac{1}{2}\lb 1 + \mu\lb \vartheta_{d_B}\circ P\rb\rb < p_1 < \frac{1}{2}\hspace*{0.5cm}\text{ and }\hspace*{0.5cm} \frac{1}{2}\lb 1 + \mu\lb P\rb\rb < p_2 < \frac{1}{2}
\end{equation}
and any $n\in\N$ the mixed tensor power of Werner maps
\[
W^{\otimes n}_{p_1}\otimes (\vartheta_{d_A}\circ W_{p_2})^{\otimes n}:\M^{\otimes 2n}_{d_A}\ra\M^{\otimes 2n}_{d_A},
\]
is decomposable. Now note that the Choi matrix of this map is given by 
\[
H_{n}:=\rho_W(p_1)^{\otimes n}\otimes (\rho_W(p_2)^\Gamma)^{\otimes n}.
\]
Since $p_1,p_2<1/2$ the states $\rho_W(p_1)$ and $\rho_W(p_2)$ satisfy $\rho_W(p_1)^\Gamma\ngeq 0$ and $\rho_W(p_2)^\Gamma\ngeq 0$ we can apply Theorem \ref{thm:DistPPTpres} to obtain PPT-preserving operations $L_n,L'_n:\M_{d^n_A}\otimes \M_{d^n_B}\ra\M_2\otimes \M_2$ of the form \eqref{equ:DistMapsPPT} such that 
\[
\lim_{n\ra \infty} L_n\lb \rho_W(p_1)^{\otimes n}\rb = \lim_{n\ra \infty} L'_n\lb \rho_W(p_2)^{\otimes n}\rb = \omega_2/2.
\]
This shows that for the sequence $\tilde{L}_n = (\ident_2\otimes \vartheta_2)\circ L'_n\circ(\ident_{d_A}\otimes \vartheta_{d_A})$ we have
\[
\lim_{n\ra\infty} (L_n\otimes \tilde{L}_n)(H_{n}) = \omega_2/2\otimes \mathbbm{F}_2/2 = C_{Q}/4.
\]
Using the form \eqref{equ:DistMapsPPT} of the maps $L_n$ and $L'_n$ it is straightforward to check that for every $n\in\N$ the map $(L_n\otimes \tilde{L}_n)$ is PPT-preserving. Applying Lemma \ref{lem:PPTPresAndDecomp} and since the set of decomposable maps is closed, we conclude that $Q:\M_{4}\ra\M_4$ must be decomposable. However, using the Choi-Jamiolkowski isomorphism it is easy to see that $Q = \ident_2\otimes \vartheta_2$, which is not even a positive map. This is a contradiction and finishes the proof. 
\end{proof}

Note that the previous proof is \emph{not} quantitative and it does \emph{not} give an explicit bound as in Theorem \ref{thm:QuantBound} on the number of tensor powers that a given decomposable map can stay decomposable for. Such a bound could probably be obtained from a more careful analysis of Theorem \ref{thm:DistPPTpres}, but we did not try this. 

\section{Applications}
\label{sec:Appl}

\subsection{Non-decomposable positive maps from tensor powers}

To construct non-decomposable positive maps using Theorem \ref{thm:LinProgrForDec} (or Theorem \ref{thm:AnalyticPBound}) we need to determine when a map of the form $W^{\otimes n}_{p_1}\otimes (\vartheta_{d}\circ W_{p_2})^{\otimes m}$ is positive. In general, this is not an easy task. In \cite{lami2016bipartite} it has been shown that $W_{p_1}\otimes (\vartheta_{d}\circ W_{p_2})$ is positive iff the completely positive map $W_{p_1}\otimes W_{p_2}$ is entanglement annihilating, i.e.~$(W_{p_1}\otimes W_{p_2})\lb X\rb$ is separable for any $X\in (\M_{d}\otimes \M_d)^+$. Moreover, the authors determined the region of parameters $p_1,p_2\in \lbr 0,1\rbr$ for which this is true. 

To state these results, it is convenient to first consider a different parametrization of the Werner maps $W_{p}$. For a parameter $t\in \lbr -1,1\rbr$ consider the map $Z_t:\M_d\ra\M_d$ given by
\[
Z_t(X) = \text{Tr}(X)\one_d - tX^T
\] 
for $X\in \M_d$. Now note that 
\begin{equation}
W_p = \frac{d+1-2p}{d(d^2-1)}Z_{t_p}
\label{equ:pVstParam}
\end{equation}
with 
\begin{equation}
t_p := \frac{d+1-2pd}{d+1-2p}.
\label{equ:tinp}
\end{equation}
Using \cite[Theorem 7]{lami2016bipartite} we find that the map $Z_{t_1}\otimes Z_{t_2}$ is entanglement annihilating (and equivalently $Z_{t_1}\otimes (\vartheta_d\circ Z_{t_2})$ is positive) iff 
\begin{equation}
0\leq -t_1t_2 - (t_1 + t_2) +2.
\label{equ:CritZZEA}
\end{equation}
Changing parametrization, this leads to a condition for positivity of the map $W_{p_1}\otimes (\vartheta_d\circ W_{p_2})$. 

Starting with the special case $t=t_1=t_2$, we find that \eqref{equ:CritZZEA} is equivalent to $t\leq \sqrt{3}-1$. Now, using \eqref{equ:tinp} shows that $W_{p}\otimes (\vartheta_d\circ W_{p})$ positive iff 
\[
p\geq \frac{(2-\sqrt{3})(d+1)}{2(d+1-\sqrt{3})}.
\]
Using Corollary \ref{cor:n1m1Equalp} we obtain:

\begin{thm}[Non-decomposable positive maps from tensor products I]
For $d>2$ consider the Werner map $W_p:\M_d\ra\M_d$ as defined in Section \ref{sec:WernerStatesAndMaps}. The map $W_{p}\otimes (\vartheta_d\circ W_{p})$ is positive and non-decomposable iff 
\begin{equation}
\frac{(2-\sqrt{3})(d+1)}{2(d+1-\sqrt{3})}\leq p <\frac{1}{2+\sqrt{\frac{2d}{d+1}}}.
\label{equ:IntervalWWND}
\end{equation}
\end{thm}

Note that the interval in \eqref{equ:IntervalWWND} is empty when $d=2$. This can be seen as further evidence for the conjecture of~\cite{filippov2017positive} that for any positive map $P:\M_2\ra\M_2$ the tensor power $P\otimes P$ is either non-positive or decomposable. 

For general $p_1,p_2\in\lbr 0,1\rbr$ we can simply insert \eqref{equ:tinp} into \eqref{equ:CritZZEA}. Together with Theorem \ref{thm:n1m1} we obtain:

\begin{thm}[Non-decomposable positive maps from tensor products II]
For $d>2$ and $p_1,p_2\in\lbr 0,1\rbr$ consider the Werner maps $W_{p_1},W_{p_2}:\M_d\ra\M_d$ as defined in Section \ref{sec:WernerStatesAndMaps}. The map $W_{p_1}\otimes (\vartheta_d\circ W_{p_2})$ is positive and non-decomposable iff
\[
0\leq -\frac{(d+1-2p_1d)(d+1-2p_2d)}{(d+1-2p_1)(d+1-2p_2)} - \lb\frac{d+1-2p_1d}{d+1-2p_1} + \frac{d+1-2p_2d}{d+1-2p_2}\rb +2,
\]
and 
\[
\frac{d-1}{d+1}p_1p_2 + p_1(1-p_2) + (1-p_1)p_2 -(1-p_1)(1-p_2)<0.
\]
\label{thm:NonDecFromTPow}
\end{thm}

\begin{figure*}[t!]
    \centering
    \begin{subfigure}[l]{0.5\textwidth}
        \centering
        \includegraphics[scale=0.6]{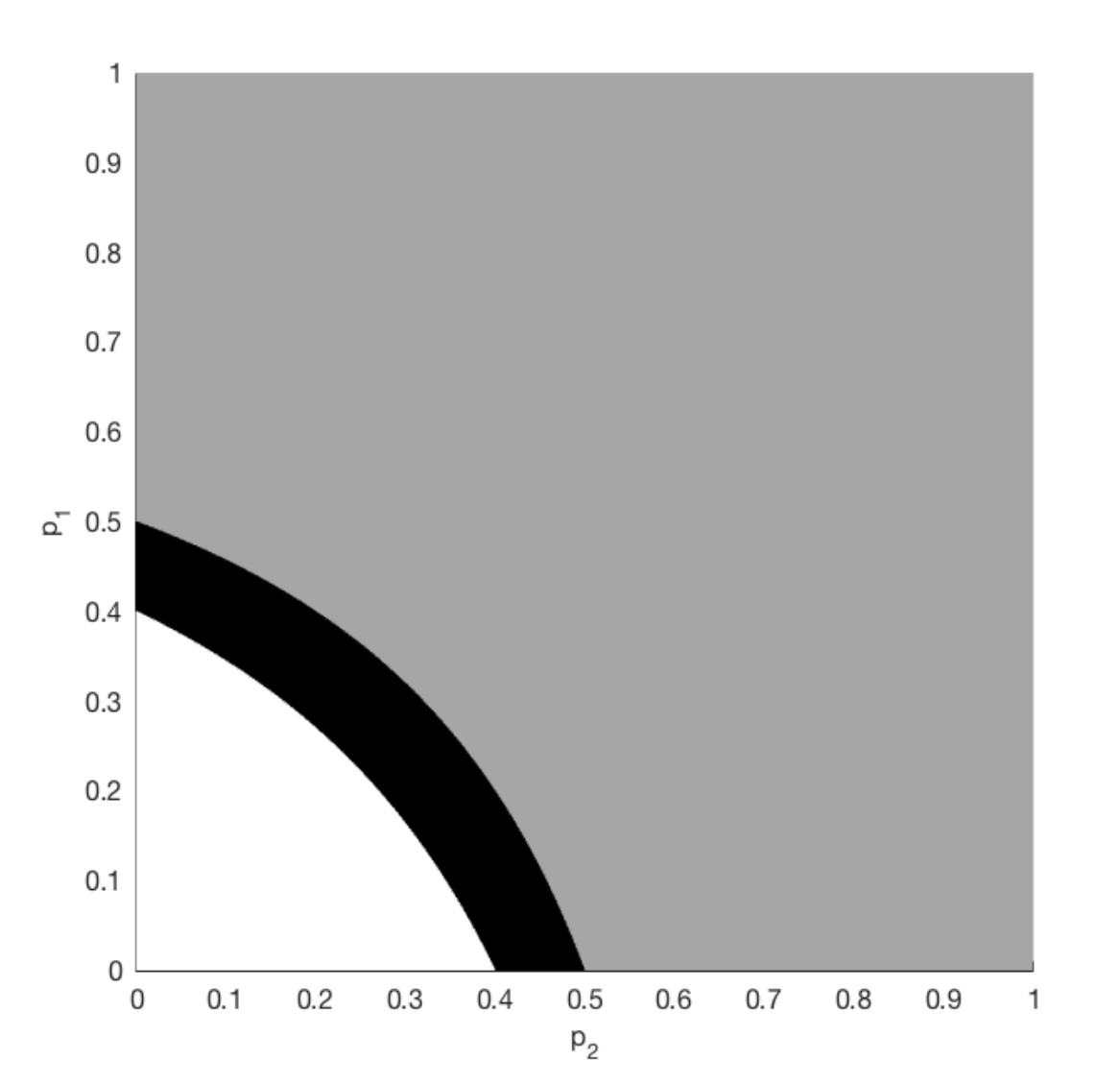}
        \caption{$d=3$}
    \end{subfigure}%
    ~ 
    \begin{subfigure}[r]{0.5\textwidth}
        \centering
        \includegraphics[scale=0.6]{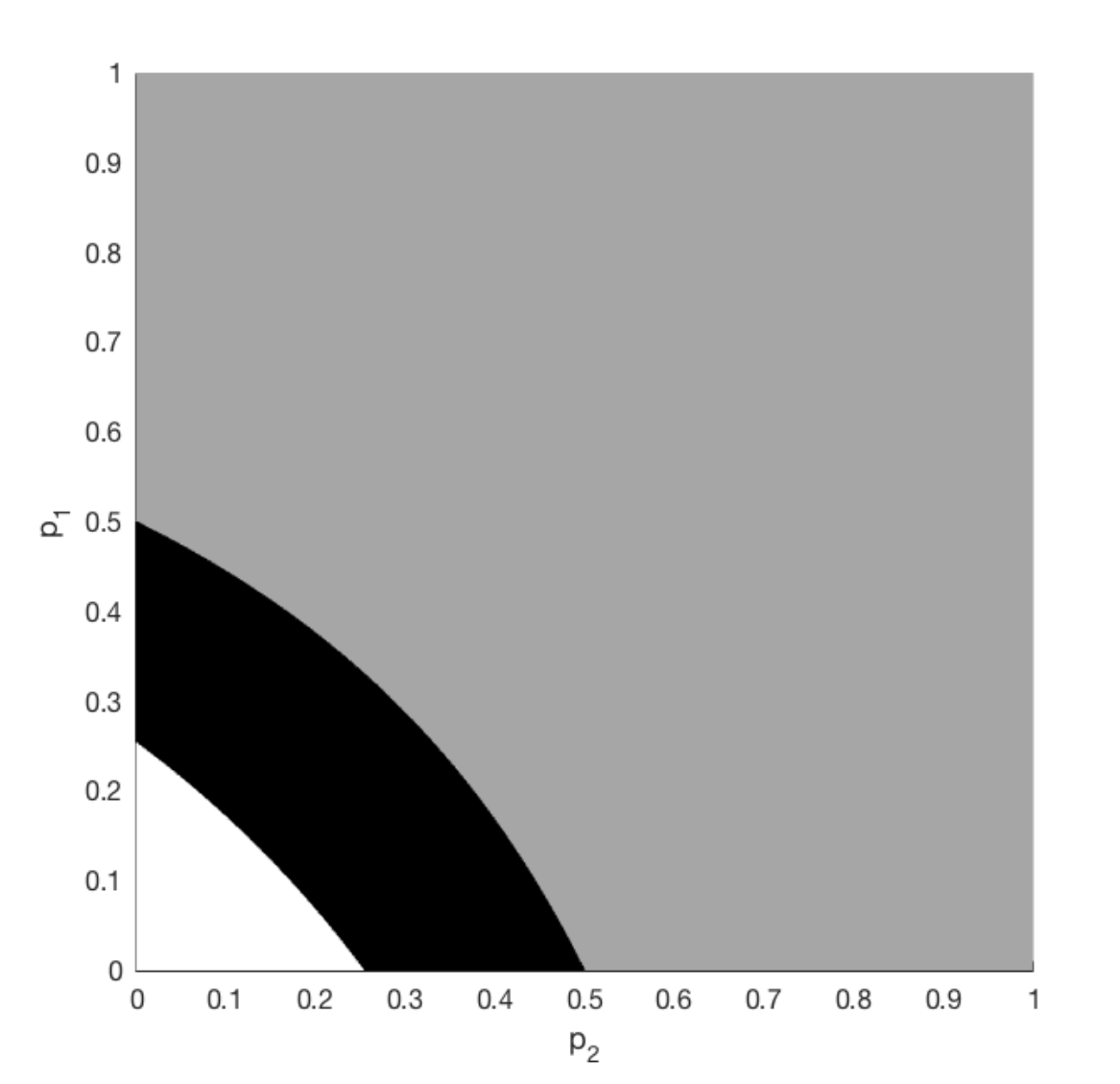}
        \caption{$d=100$}
    \end{subfigure}
    \caption{The parameters $(p_1,p_2)$ where $W_{p_1}\otimes (\vartheta_d\circ W_{p_2})$ is positive and non-decomposable are in the interior of the black region. The gray region contains the remaining parameters where the map is positive.}
    \label{fig:nonDec}
\end{figure*}

In Fig.~\ref{fig:nonDec} we plot the parameter region where the linear map $W_{p_1}\otimes (\vartheta_d\circ W_{p_2})$ is positive and non-decomposable in the cases $d=3$ and $d=100$.

\subsection{Implications for the PPT squared conjecture}

The PPT squared conjecture \cite{christandl2012PPT} asks for two completely positive and completely copositive maps $T_1, T_2 :\M_{d}\ra \M_{d}$ whether their composition $T_2\circ T_1$ is entanglement breaking. Here, PPT is short for ``positive partial transposition'' and linear maps that are both completely positive and completely copositive are often called PPT binding maps in the context of quantum information theory. It has been shown in \cite{muller2018PPT} that this conjecture is equivalent to 

\begin{conj}[Equivalent formulation of PPT squared]
For any completely positive and completely copositive map $T:\M_{d_1}\ra\M_{d_2}$ and any positive map $P:\M_{d_2}\ra\M_{d_3}$ the composition $P\circ T$ is decomposable.
\label{conj:PPTSquare}
\end{conj}

Here we will prove: 

\begin{thm}
If there exists a non-trivial tensor-stable positive map, then Conjecture \ref{conj:PPTSquare} is false.
\end{thm}

\begin{proof}
Let $P:\M_{d_1}\ra\M_{d_2}$ be a non-trivial tensor-stable positive map. We can assume without loss of generality that $P$ is not decomposable. Otherwise, Theorem \ref{thm:NoTensStabDec} implies that for some $n\in \N$ the map $P^{\otimes n}$ is not decomposable, but tensor-stable positive, and we could consider it instead of $P$. 

Since $P$ is not decomposable, there exists a completely positive and completely copositive map $T:\M_{d_2}\ra\M_{d_1}$ such that $P\circ T$ is not completely positive. Let $D:\M_2\ra \M_2$ be given by
\[
D(X)=\bra{0}X\ket{0}\proj{0}{0} +\bra{1}X\ket{1}\proj{1}{1},
\] 
for any $X\in\M_2$, and note that this map is entanglement breaking~\cite{horodecki2003entanglement}. Now consider the modified maps $\tilde{P}:\M_{d_1}\otimes \M_2\ra\M_{d_2}\otimes \M_2$ given by $\tilde{P} = P\otimes D$ and $\tilde{T}:\M_{d_2}\ra\M_{d_1}\otimes \M_2 $ given by 
\begin{align*}
\tilde{T} = T\otimes\proj{0}{0} + T\circ \vartheta_{d_1}\otimes \proj{1}{1}.  
\end{align*}
Note that $\tilde{T}$ is still completely positive and completely copositive. The map $\tilde{P}$ is still tensor-stable positive, since by permuting subsystems we have $\tilde{P}^{\otimes n} \simeq P^{\otimes n}\otimes D^{\otimes n}$ which is positive since $P^{\otimes n}$ is positive and $D^{\otimes n}$ is entanglement breaking. Furthermore, observe that 
\[
\tilde{P}\circ \tilde{T} = P\circ T\otimes\proj{0}{0} + P\circ T\circ \vartheta_{d_1}\otimes \proj{1}{1}
\]
is neither completely positive nor completely co-positive, since $P\circ T$ is not completely positive. Therefore, we can use Theorem \ref{thm:NoTensStabDec} to find an $n\in\N$ such that $(\tilde{P}\circ \tilde{T})^{\otimes n}= \tilde{P}^{\otimes n}\circ \tilde{T}^{\otimes n}$ is not decomposable. But then $\tilde{T}^{\otimes n}$ and $\tilde{P}^{\otimes n}$ provide a counterexample for Conjecture \ref{conj:PPTSquare}.
\end{proof}

\section{Conclusion and open problems}

We have shown that a tensor-stable decomposable map has to be either completely positive or completely copositive. Moreover, we have derived bounds on the number of tensor powers a given linear map can stay decomposable for, if it does not belong to these two classes. Finally, we applied our results to construct non-decomposable positive maps, and to find a link between the existence problem of tensor-stable positive maps and the PPT squared conjecture in quantum information theory.    

It would be interesting for future research to improve the quantitative bound from Theorem \ref{thm:QuantBound} and possibly to find optimal bounds on the decomposability of tensor powers of positive maps. For the mixed tensor powers of Werner maps, studied throughout our article, this could be possible by studying the set of quantum states with positive partial transpose introduced in Section \ref{subsec:SymmStates}.

Finally, it should also be noted that the results from Section \ref{sec:Symmetrizing} can be applied to similar questions for other mapping cones (as in Definition \ref{defn:MappingCone}): Given a mapping cone $\mathcal{C}$ which elements $P\in\mathcal{C}$ satisfy $P^{\otimes n}\in\mathcal{C}$ for all $n\in\N$? A mapping cone might contain the trivial examples of completely positive or completely copositive maps, but finding any example not belonging to these classes would also be an example of a non-trivial tensor-stable positive map (since mapping cones are always subcones of positive maps). Following the lines of the second proof of our Theorem \ref{thm:NoTensStabDec} from Section \ref{sec:Proofs} one could try to show that there are no such maps (except possibly the trivial ones) in a given mapping cone. To do this one would need to understand the properties of the linear maps leaving invariant the Choi matrices of the mapping cone in question. We will study this further in future works.

\section*{Acknowledgments}
We would like to thank Chris Perry for proofreading this manuscript. We acknowledge financial support from the European Research Council (ERC Grant Agreement no 337603), the Danish Council for Independent Research (Sapere Aude) and VILLUM FONDEN via the QMATH Centre of Excellence (Grant No. 10059).

\appendix

\section{Proof of Lemma \ref{lem:ExtremePoints}}
\label{sec:Appendix}

Lemma \ref{lem:ExtremePoints} has been shown originally in \cite{vollbrecht2002activating}. We include here a proof of this result to make the presentation selfcontained. Let $\mathcal{C}_d$ denote the set from \eqref{equ:Cset}. Using the condition $\sum^2_{i,j=1}Q_{ij} = 1$ we can write each $Q\in \mathcal{C}_d$ as 
\[
Q = \begin{pmatrix} x & y \\ 1-x-y-z & z \end{pmatrix}.
\]
with parameters $x,y,z\in\R$. With these parameters, the conditions that $QV^{1}_d$ and $Q^TV^{1}_d$ are entrywise positive, are equivalent to the set of inequalities
\begin{align}
x-y&\geq 0, \label{equ:inequ1}\\
x + \alpha_d y&\geq 0,\label{equ:inequ2}\\
y - z&\geq 0,\label{equ:inequ3}\\
y + \alpha_d z&\geq 0,\label{equ:inequ4}\\
1-x-y - 2z&\geq 0,\label{equ:inequ5}\\
1-x-y+(\alpha_d-1)z&\geq 0,\label{equ:inequ6}\\
2x + y + z &\geq 1, \label{equ:inequ7}\\
-(\alpha_d - 1)x - \alpha_d( y + z) &\geq -\alpha_d. \label{equ:inequ8}
\end{align}
Here, $\alpha_d = \frac{d+1}{d-1}$ and in particular we have $\alpha_d>1$ for any $d\in\N$. Let $\mathcal{C}'_d$ denote the set of all $(x,y,z)\in \R^3$ satisfying the above set of inequalities. Clearly, to prove Lemma \ref{lem:ExtremePoints} it is enough to show that the extreme points of $\mathcal{C}'_d$ are given by 
\begin{align*}
Q'_1 &= \frac{1}{2(d+2)}(d+1,d+1,-(d-1)),\\
Q'_2 &= (1,0,0),\\
Q'_3 &= \frac{1}{4}(1,1,1),\\
Q'_4 &= \frac{1}{2}(1,0,0),\\
Q'_5 &= \frac{1}{2}(1,1,0).
\end{align*}
We will show this in the following two lemmas:

\begin{lem}
The points $Q'_1, Q'_2, Q'_3, Q'_4$, and $Q'_5$ defined above are extreme points of $\mathcal{C}'_d$.
\end{lem}

\begin{proof}
Consider first $(x,y,z)\in \mathcal{C}'_d$ satisfying $x=y$. Using \eqref{equ:inequ6} and \eqref{equ:inequ4} we find that
\[
1+(3\alpha_d -1)z = (1-2x+(\alpha_d-1)z) + 2(x + \alpha_d z) \geq 0
\]  
showing that 
\begin{equation}
z \geq -\frac{1}{3\alpha_d -1} = -\frac{d-1}{2(d+2)}.
\label{equ:Adv1}
\end{equation}
From \eqref{equ:inequ3} and \eqref{equ:inequ7} we have that 
\begin{equation}
x\geq \frac{1}{4}.  
\label{equ:Adv2}
\end{equation}
Finally, assuming $z\geq 0$ and using \eqref{equ:inequ5} or assuming $z<0$ and using \eqref{equ:inequ6} implies that 
\begin{equation}
x\leq \frac{1}{2}.
\label{equ:Adv3}
\end{equation}
Now we can show that the points $Q_1,\ldots , Q_5$ are indeed extreme points of the set $\mathcal{C}'_d$.
\begin{itemize}
\item Assume that there are $(x_1,y_1,z_1),(x_2,y_2,z_2)\in \mathcal{C}'_d$ such that 
\[
Q'_1 = p(x_1,y_1,z_1) + (1-p)(x_2,y_2,z_2) 
\]
for some $p\in \lb 0,1\rb$. Since the first two entries of $Q'_1$ are equal we find
\[
p(x_1-y_1) + (1-p)(x_2-y_2) = 0,
\]
which using \eqref{equ:inequ1} implies that $x_1=y_1$ and $x_2=y_2$. Clearly, \eqref{equ:Adv1} implies that $z_1=z_2 = -\frac{d-1}{2(d+2)}$. Inserting these in \eqref{equ:inequ7} and using that $x_1=y_1$ and $x_2=y_2$ shows that $x_1=y_1=x_2=y_2=\frac{d+1}{2(d+2)}$, and thus $Q'_1$ is an extreme point of $\mathcal{C}'_d$.

\item If for $(x_1,y_1,z_1),(x_2,y_2,z_2)\in \mathcal{C}'_d$ we have
\[
Q'_3 = p(x_1,y_1,z_1) + (1-p)(x_2,y_2,z_2) 
\]
for some $p\in \lb 0,1\rb$, then we can argue as for $Q'_1$ to conclude that $x_1=y_1$ and $x_2=y_2$. Now, \eqref{equ:Adv2} implies that $x_1=x_2 =y_1=y_2= \frac{1}{4}$. Then, by \eqref{equ:inequ5} we find that $x_3=y_3=\frac{1}{4}$ as well, and $Q'_3$ is indeed an extreme point of $\mathcal{C}'_d$.

\item If for $(x_1,y_1,z_1),(x_2,y_2,z_2)\in \mathcal{C}'_d$ we have
\[
Q'_5 = p(x_1,y_1,z_1) + (1-p)(x_2,y_2,z_2) 
\]
for some $p\in \lb 0,1\rb$, then we can argue as for $Q'_1$ to conclude that $x_1=y_1$ and $x_2=y_2$. Now, \eqref{equ:Adv3} implies that $x_1=x_2=y_1=y_2=\frac{1}{2}$, and using \eqref{equ:inequ5} and \eqref{equ:inequ6} we find that $z_1=z_2=0$. Therefore, $Q'_5$ is also an extreme point of $\mathcal{C}'_d$.

\item Let $(x_1,y_1,z_1),(x_2,y_2,z_2)\in \mathcal{C}'_d$ be such that 
\[
Q'_4 = p(x_1,y_1,z_1) + (1-p)(x_2,y_2,z_2) 
\]
for some $p\in \lb 0,1\rb$. Since the last two entries of $Q'_4$ are equal we find
\[
p(y_1-z_1) + (1-p)(y_2-z_2) = 0,
\]
which using \eqref{equ:inequ3} implies that $y_1=z_1$ and $y_2=z_2$. By \eqref{equ:inequ4} and since the last two entries of $Q'_4$ vanish we need to have $y_1=z_1=y_2=z_2=0$. Using \eqref{equ:inequ7} we find that $x_1=x_2=\frac{1}{2}$, which shows that $Q'_4$ is an extreme point of $\mathcal{C}'_d$. 

\item To prove extremality of $Q'_2$ one more observations. First, note that if for $(x,y,z)\in \mathcal{C}'_d$ we would have $x>1$, then by \eqref{equ:inequ5} and \eqref{equ:inequ3} we have $z<0$. Then, \eqref{equ:inequ4} implies that $y>0$, which leads to a contradiction in \eqref{equ:inequ6}. Therefore, we have 
\begin{equation}
x\leq 1.
\label{equ:Adv4}
\end{equation}

Finally, let $(x_1,y_1,z_1),(x_2,y_2,z_2)\in \mathcal{C}'_d$ such that
\[
Q'_2 = p(x_1,y_1,z_1) + (1-p)(x_2,y_2,z_2) 
\]
for some $p\in \lb 0,1\rb$. Then, by \eqref{equ:Adv4} we have $x_1=x_2=1$. Arguing as for $Q'_4$ (using that the final two entries of $Q'_2$ coincide) we also have $y_1=z_1$ and $y_2=z_2$. Now, \eqref{equ:inequ4} and \eqref{equ:inequ5} show that $y_1=z_1=y_2=z_2=0$ proving that $Q'_2$ is an extreme point of $\mathcal{C}'_d$.

\end{itemize}

\end{proof}

We still need to show the following:

\begin{lem}
The points $Q'_1,Q'_2,Q'_3,Q'_4,Q'_5$ are all the extreme points of $\mathcal{C}'_d$.
\end{lem} 

\begin{proof}
For this consider the plane 
\begin{equation}
\mathcal{S}:=\lset (x,y,z)\in \R^3~:~y = \frac{1}{2}(1-x-z)\rset.
\label{equ:Splane}
\end{equation} 
Inserting $y = \frac{1}{2}(1-x-z)$ into the inequalities \eqref{equ:inequ1},\eqref{equ:inequ2},\ldots ,\eqref{equ:inequ8} leads to the set of inequalities 
\begin{align}
3x+z&\geq 1, \label{equ:inequ9}\\
x + 3z &\leq 1,\label{equ:inequ10}\\
(2-\alpha_d)x - \alpha_d z &\geq -\alpha_d. \label{equ:inequ11}\\
-x + (2\alpha_d -1)z &\geq -1. \label{equ:inequ12}
\end{align} 
characterizing the $x,z\in \R$ such that $(x,\frac{1}{2}(1-x-z),z)\in \mathcal{S}\cap \mathcal{C}_d'$. By computing the intersection points of the lines given by the equality cases of the inequalities \eqref{equ:inequ9},\eqref{equ:inequ10} and \eqref{equ:inequ12} it is easy to see that the extreme points $\mathcal{S}\cap \mathcal{C}_d'$ are $Q'_1,Q'_2$ and $Q'_3$. Note, that \eqref{equ:inequ11} turns out to be redundant for the definition of this convex set. See Fig.~\ref{fig:PlaneCutSimp} for a picture of how $\mathcal{S}$ intersects the set $\mathcal{C}_d'$.

\begin{figure*}[t!]
        \centering
        \includegraphics[scale=0.7]{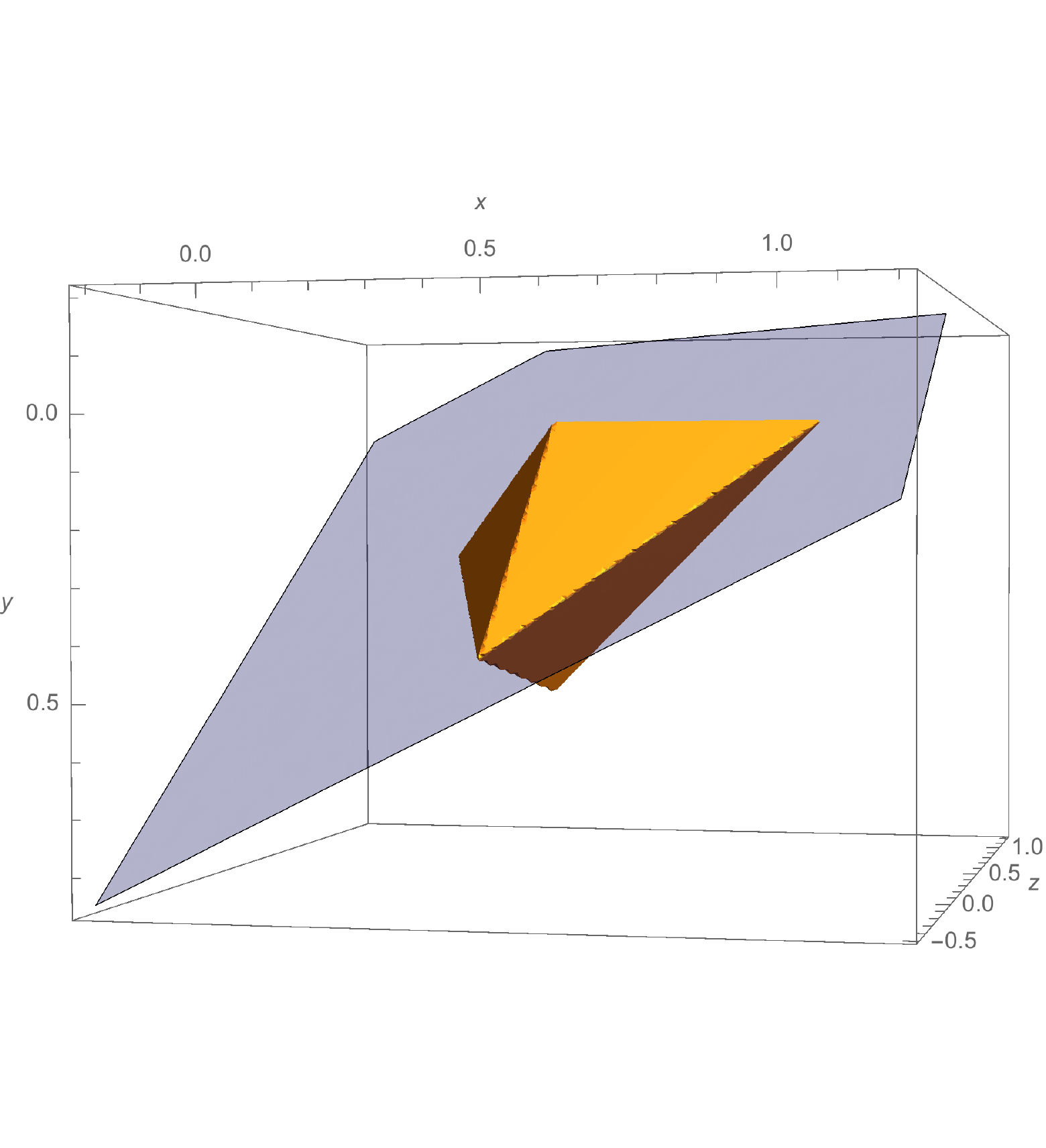}
        \caption{The set $\mathcal{C}'_5$ together with the intersecting plane $\mathcal{S}$.}
        \label{fig:PlaneCutSimp}
\end{figure*}

Let $(x,y,z)\in \mathcal{C}'_d\setminus (\lset Q'_5, Q'_4\rset\cup \mathcal{S})$ and assume first that $y>\frac{1}{2}(1-x-z)$. Now, consider
\[
\lambda := 4\lb y - \frac{1}{2}\lb 1- x-z\rb\rb >0.
\]
We will show that $\lambda < 1$ and that 
\begin{equation}
(x',y',z') = \frac{1}{1-\lambda} \lb (x,y,z) - \lambda Q'_5\rb\in \mathcal{S}\cap\mathcal{C}'_d,
\label{equ:pbar1}
\end{equation}
showing that $(x,y,z)$ is not an extreme point of $\mathcal{C}'_d$. Note that by \eqref{equ:inequ5},\eqref{equ:inequ6}, and \eqref{equ:inequ1} we have that 
\begin{equation}
x+2y + z = \frac{d+2}{2d}\lb x + y + 2z\rb + \frac{d-1}{d}\lb x + y - (\alpha_d-1)z\rb + \frac{1}{2}\lb y-x\rb\leq \frac{3}{2}.
\label{equ:inequSpec1}
\end{equation}
We have equality in \eqref{equ:inequSpec1} iff we have equality in \eqref{equ:inequ5},\eqref{equ:inequ6}, and \eqref{equ:inequ1}. This is equivalent to $(x,y,z)=Q'_5$, which we excluded previously. Therefore, our point $(x,y,z)$ satisfies \eqref{equ:inequSpec1} strictly, which is easily seen to be equivalent to $\lambda <1$. Finally, it is easy to verify that $(x',y',z')$ defined in \eqref{equ:pbar1} satisfies the inequalities \eqref{equ:inequ9},\eqref{equ:inequ10},\eqref{equ:inequ11}, and \eqref{equ:inequ12} showing that  $(x',y',z')\in \mathcal{S}\cap\mathcal{C}'_d$. 

Now, let $(x,y,z)\in \mathcal{C}'_d\setminus (\lset Q'_5, Q'_4\rset\cup \mathcal{S})$ with $y<\frac{1}{2}(1-x-z)$. Similarly to the previous argument we consider
\[
\lambda := -4\lb y - \frac{1}{2}\lb 1- x-z\rb\rb >0.
\]
Again, we will show that $\lambda < 1$ and that 
\begin{equation}
(x',y',z') = \frac{1}{1-\lambda} \lb (x,y,z) - \lambda Q'_2\rb\in \mathcal{S}\cap\mathcal{C}'_d,
\label{equ:pbar12}
\end{equation}
showing that $(x,y,z)$ is not an extreme point of $\mathcal{C}'_d$. By \eqref{equ:inequ7},\eqref{equ:inequ2}, and \eqref{equ:inequ3} we have that 
\begin{equation}
x+2y + z = \frac{1}{2}\lb 2x + y + z\rb + \frac{2(d-1)}{d}\lb y + \alpha_d z\rb + \frac{d+2}{d}\lb y-z\rb\geq \frac{1}{2}.
\label{equ:inequSpec12}
\end{equation}
We have equality in \eqref{equ:inequSpec12} iff we have equality in \eqref{equ:inequ7},\eqref{equ:inequ2}, and \eqref{equ:inequ3}. This is equivalent to $(x,y,z)=Q'_4$, which we excluded previously. Therefore, our point $(x,y,z)$ satisfies \eqref{equ:inequSpec12} strictly, which is easily seen to be equivalent to $\lambda <1$. Finally, it is easy to verify that $(x',y',z')$ defined in \eqref{equ:pbar12} satisfies the inequalities \eqref{equ:inequ9},\eqref{equ:inequ10},\eqref{equ:inequ11}, and \eqref{equ:inequ12} showing that  $(x',y',z')\in \mathcal{S}\cap\mathcal{C}'_d$. 

The above argument shows that there are no further extreme points $(x,y,z)\in \mathcal{C}'_d\setminus (\lset Q'_5, Q'_4\rset\cup \mathcal{S})$ of the set $\mathcal{C}'_d$. Indeed, $Q'_1,\ldots ,Q'_5$ are all the extreme points of $\mathcal{C}'_d$ finishing the proof.
\end{proof}

\bibliographystyle{IEEEtran}
\bibliography{mybibliography}

\end{document}